%% file: paper.tex
\documentclass[sigplan,10pt,nonacm]{acmart}
\setcopyright{acmlicensed}
\copyrightyear{2018}
\acmYear{2018}
\acmDOI{XXXXXXX.XXXXXXX}
\acmConference[Conference acronym 'XX]{Make sure to enter the correct
  conference title from your rights confirmation email}{June 03--05,
  2018}{Woodstock, NY}
\acmISBN{978-1-4503-XXXX-X/2018/06}

\pdfoutput=1 
\input{config}

\NewDocumentCommand{\cmuAffil}{}{\affiliation{\institution{Carnegie Mellon University}\city{Pittsburgh}\state{PA}\country{USA}}}

\begin{document}

\title{Mechanizing a Proof-Relevant Logical Relation for Timed Message-Passing Protocols}

\author{Tesla Zhang}
\email{teslaz@cmu.edu}
\orcid{https://orcid.org/0000-0002-9050-846X}
\cmuAffil

\author{Asher Kornfeld}
\email{akornfel@andrew.cmu.edu}
\orcid{https://orcid.org/0009-0008-5788-5184}
\cmuAffil

\author{Rui Li}
\email{ruil3@andrew.cmu.edu}
\orcid{https://orcid.org/0009-0006-3555-9770}
\cmuAffil

\author{Sonya Simkin}
\email{ssimkin@andrew.cmu.edu}
\orcid{https://orcid.org/0009-0008-9261-6318}
\cmuAffil

\author{Yue Yao}
\email{yueyao@cs.cmu.edu}
\orcid{https://orcid.org/0000-0001-8523-5156}
\cmuAffil

\author{Stephanie Balzer}
\email{balzers@cs.cmu.edu}
\orcid{https://orcid.org/0000-0002-8347-3529}
\cmuAffil




\begin{abstract}
\input{sections/abstract}
\end{abstract}

\maketitle

\section{Introduction}%
\label{sec:intro}
\input{sections/introduction}

\section{Semantic Verification Framework}%
\label{sec:verification-framework}
\input{sections/verification-framework}

\section{Fundamental Theorem}%
\label{sec:fundamental-theorem}
\input{sections/fundamental-theorem}

\section{Mechanization}%
\label{sec:mechanization}
\input{sections/mechanization}

\section{Related Work}%
\label{sec:related}
\input{sections/related}

\begin{acks}
\input{sections/acknowledge.tex}
\end{acks}

\bibliographystyle{ACM-Reference-Format}
\bibliography{../paper/ref}

\end{document}

%% file: config.tex
\usepackage{tiz-note}
\usepackage{utfsym}
\usepackage{enumitem}
\setlist{leftmargin=\parindent}

\definecolor{AyaFn}{HTML}{00627a}
\definecolor{AyaConstructor}{HTML}{067d17}
\definecolor{AyaStruct}{HTML}{00627a}
\definecolor{AyaGeneralized}{HTML}{00627a}
\definecolor{AyaData}{HTML}{00627a}
\definecolor{AyaPrimitive}{HTML}{00627a}
\definecolor{AyaKeyword}{HTML}{0033b3}
\definecolor{AyaComment}{HTML}{8c8c8c}
\definecolor{AyaField}{HTML}{871094}

\NewDocumentCommand{\stepsTo}{mm}{\xmapsto[#2]{~#1~}}
\NewDocumentCommand{\stepsToObj}{mm}{\mathrel{\mapstochar\relbar\mathrel{\mkern-14mu}\stepsTo{#1}{#2}}}
\NewDocumentCommand{\CfgStepsOperad}{mm}{{\usym{1F551}}_{#2},{#1}}
\NewDocumentCommand{\CfgSteps}{}{\mapsto}
\NewDocumentCommand{\CfgRefl}{mm}{\mathsf{refl}^{#1}_{#2}}
\NewDocumentCommand{\CfgStepT}{mmmm}{\operatorname{\mathsf{stepT}}^{#1}_{#2;#3}(#4)}
\NewDocumentCommand{\CfgStepC}{mmmm}{\operatorname{\mathsf{stepC}}^{#1;#2}_{#3}(#4)}
\NewDocumentCommand{\CfgStepsFrame}{mm}{\operatorname{\mathsf{stepF}}^{#1}(#2)}
\NewDocumentCommand{\CfgStepsConcat}{mm}{\operatorname{\mathsf{stepCat}}(#1, #2)}
\NewDocumentCommand{\CfgStepTR}{mmmm}{\operatorname{\mathsf{stepT'}}^{#1}_{#2;#3}(#4)}
\NewDocumentCommand{\CfgStepsInterleave}{mm}{\operatorname{\mathsf{stepI}}(#1, #2)}
\NewDocumentCommand{\CfgStepsInterleaveL}{mm}{\operatorname{\mathsf{stepIL}}(#1, #2)}
\NewDocumentCommand{\CfgStepsInterleaveR}{mm}{\operatorname{\mathsf{stepIR}}(#1, #2)}
\NewDocumentCommand{\extendTraj}{mmm}{\operatorname{\mathsf{Ext}}^{#1}_{#2}(#3)}

\NewDocumentCommand{\trajTy}{mm}{\mathsf{C}(#1, #2)}
\NewDocumentCommand{\concatTraj}{mm}{\operatorname{\mathsf{cat}}(#1, #2)}
\NewDocumentCommand{\constTraj}{mmm}{\operatorname{\mathsf{c}}_{[#2, #3)}(#1)}
\NewDocumentCommand{\interleaveTraj}{mm}{\operatorname{\mathsf{trajInter}}(#1, #2)}
\NewDocumentCommand{\partBeforeTraj}{mm}{\operatorname{\mathsf{trajB}}_{#2}(#1)}
\NewDocumentCommand{\partAfterTraj}{mm}{\operatorname{\mathsf{trajA}}_{#2}(#1)}

\NewDocumentCommand{\RRefl}{m}{\operatorname{\mathsf{R\_refl}}(#1)}
\NewDocumentCommand{\RStepT}{mm}{\operatorname{\mathsf{R\_stepT}}(#1, #2)}
\NewDocumentCommand{\RStepC}{mm}{\operatorname{\mathsf{R\_stepC}}(#1, #2)}
\NewDocumentCommand{\RFrame}{mm}{\operatorname{\mathsf{R\_frame}}(#1, #2)}
\NewDocumentCommand{\RConcat}{mm}{\operatorname{\mathsf{R\_concat}}(#1, #2)}
\NewDocumentCommand{\RPartitionA}{mm}{\operatorname{\mathsf{R\_partA}}(#1, #2)}
\NewDocumentCommand{\RPartitionB}{mm}{\operatorname{\mathsf{R\_partB}}(#1, #2)}
\NewDocumentCommand{\RInterleave}{mm}{\operatorname{\mathsf{R\_inter}}(#1, #2)}
\NewDocumentCommand{\RInterleaveL}{mm}{\operatorname{\mathsf{R\_interL}}(#1, #2)}
\NewDocumentCommand{\RInterleaveR}{mm}{\operatorname{\mathsf{R\_interR}}(#1, #2)}

\NewDocumentCommand{\SRefl}{m}{\mathsf{S\_refl}(#1)}
\NewDocumentCommand{\SPartitionA}{mm}{\operatorname{\mathsf{S\_partA}}(#1, #2)}
\NewDocumentCommand{\SPartitionB}{mm}{\operatorname{\mathsf{S\_partB}}(#1, #2)}
\NewDocumentCommand{\SConcat}{mm}{\operatorname{\mathsf{S\_concat}}(#1, #2)}
\NewDocumentCommand{\SStep}{m}{\operatorname{\mathsf{S\_step}}(#1)}
\NewDocumentCommand{\SFrame}{mm}{\operatorname{\mathsf{S\_frame}}(#1, #2)}
\NewDocumentCommand{\SInterleave}{mmm}{\operatorname{\mathsf{S\_inter}}_{#3}(#1, #2)}
\NewDocumentCommand{\SCommTP}{mm}{\operatorname{\mathsf{S\_down}}(#1, #2)}
\NewDocumentCommand{\SCommTC}{mm}{\operatorname{\mathsf{S\_up}}(#1, #2)}

\NewDocumentCommand{\substOf}{m}{\underline{#1}}
\NewDocumentCommand{\typedCompl}{mm}{\operatorname{\mathsf{typedCompl}}(#1, #2)}
\NewDocumentCommand{\SInterleaveCompl}{mm}{\operatorname{\mathsf{S\_int\_compl}}(#1, #2)}

\NewDocumentCommand{\compTraj}{mm}{\operatorname{S}_{[#1, #2)}}
\NewDocumentCommand{\projStart}{}{\operatorname{\pi_{\text{start}}}}
\NewDocumentCommand{\projEnd}{}{\operatorname{\pi_{\text{end}}}}
\NewDocumentCommand{\projNtraj}{}{\operatorname{\pi_{\text{traj}}}}
\NewDocumentCommand{\projStepTo}{}{\operatorname{\pi_{\text{stepTo}}}}
\NewDocumentCommand{\projSteps}{}{\operatorname{\pi_{\text{σ}}}}
\NewDocumentCommand{\projRealize}{}{\operatorname{\pi_{\mathcal{R}}}}
\NewDocumentCommand{\isType}{m}{{#1}~\textsf{type}}
\NewDocumentCommand{\withStar}{}{\text{left}}
\NewDocumentCommand{\withoutStar}{}{\text{right}}
\NewDocumentCommand{\propOf}{}{\operatorname{\mathsf{propOf}}}
\NewDocumentCommand{\starInvert}{m}{{#1}^{-1}}
\NewDocumentCommand{\functorMap}{m}{\operatorname{\mathtt{fmap}}_{#1}}

\NewDocumentCommand{\relAt}{}{\mathrel{@}}
\NewDocumentCommand{\GF}{}{\mathcal{G};\mathcal{F}}
\NewDocumentCommand{\GtFp}{}{\mathcal{G},t;\mathcal{F},p}
\NewDocumentCommand{\GtFq}{}{\mathcal{G},t;\mathcal{F},q}

\NewDocumentCommand{\recvEx}{}{ \operatorname{\mathsf{recv}} }
\NewDocumentCommand{\sendEx}{}{ \operatorname{\mathsf{send}} }

\NewDocumentCommand{\procEx}{mm}{ \textsf{proc}[#1](#2) }
\NewDocumentCommand{\procFwd}{mm}{ \textsf{fwd}[#1](#2) }
\NewDocumentCommand{\AtomicObjPrefix}{}{ \textsf{Atomic} }
\NewDocumentCommand{\AtomicObj}{m}{ \AtomicObjPrefix(#1) }

\NewDocumentCommand{\lwith}{}{\mathrel{\&}}
\NewDocumentCommand{\fwdEx}{}{ \operatorname{\mathsf{fwd}} }
\NewDocumentCommand{\letEx}{}{ \operatorname{\mathsf{let}} }
\NewDocumentCommand{\partialmap}{}{\rightharpoonup}
\NewDocumentCommand{\KwProp}{}{\textsf{Prop}}
\NewDocumentCommand{\applyCompl}{mm}{ \operatorname{\textsf{applyCompl}}(#1, #2) }

\NewDocumentCommand{\valueInterp}{mmm}{\mathcal{V}\lrbracket{#1}^{#3}\mathrel{@}{#2}}
\NewDocumentCommand{\termInterp}{mmm}{\mathcal{E}\lrbracket{#1}^{#3}\mathrel{@}{#2}}
\NewDocumentCommand{\valueInterpYue}{mmm}{\mathcal{V}\lrbracket{#1}^{#3}\mathrel{@}{#2}}
\NewDocumentCommand{\termInterpYue}{mmm}{\mathcal{L}\lrbracket{#1}^{#3}\mathrel{@}{#2}}

\NewDocumentCommand{\Cfg}{}{\mathsf{Cfg}}
\NewDocumentCommand{\NCfg}{}{\mathsf{NCfg}}
\NewDocumentCommand{\Obj}{}{\mathsf{Obj}}
\NewDocumentCommand{\NObj}{}{\mathsf{NObj}}
\NewDocumentCommand{\Labels}{}{\mathbb{A}}
\NewDocumentCommand{\FinTimeT}{}{\mathsf{time}_{\mathsf{fin}}}
\NewDocumentCommand{\TimeT}{}{\mathsf{time}}
\NewDocumentCommand{\FMSet}{m}{\mathsf{FM}(#1)}

\NewDocumentCommand{\Actions}{}{\mathsf{Act}}
\NewDocumentCommand{\fmconcat}{}{\mathrel{⊎}}

\newcommand{\tagthisline}{%
  \refstepcounter{equation}%
  (\theequation)%
}

\usepackage[auto]{contour}
\NewDocumentCommand{\nameless}{m}{\contour{black}{\color{white}{#1}}}

\NewDocumentCommand{\textcode}{m}{\textsf{#1}}
\NewDocumentCommand{\textlibraryname}{m}{\texttt{#1}}

\NewDocumentCommand{\ProcLang}{}{\textcode{ProcLang}}

\RequirePackage{tikz}
\usetikzlibrary{automata, positioning, arrows}

\NewDocumentCommand{\TODO}{m m}%
  {{\bfseries\color{#1}[#2]}}%

\crefname{example}{Example}{Examples}
\Crefname{example}{Example}{Examples}



%% file: sections/abstract.tex
Semantic typing has become a powerful tool for program verification,
applying the technique of logical relations as not only a proof method,
but also a device for prescribing program behavior.
In recent work, \citeauthor{YaoPOPL2025} scaled semantic typing to
the verification of timed message-passing protocols, which are
prevalent in, for example, Internet of Things (IoT) and
real-time systems applications.
The appeal of semantic typing in this context is precisely because of its ability to
support typed and untyped program components alike---%
including physical objects---%
which caters to the heterogeneity of these applications.
Another demand inherent to these applications is timing:
constraining the time or time window within which
a message exchange must happen.
\citeauthor{YaoPOPL2025} equipped their logical relation not only with temporal predicates,
but also with computable trajectories,
to supply the evidence that an inhabitant
can step from one time point to another one.
While \citeauthor{YaoPOPL2025} provide the formalization for such a verification tool, it lacks a mechanization.
Mechanizing the system would not only provide a machine proof for it,
but also facilitate scalability for future extensions and applications.

This paper tackles the challenge of mechanizing
the resulting proof-relevant logical relation in a proof assistant.
allowing trajectories to be interleaved, partitioned, and concatenated,
while the intended equality on trajectories is the equality of their graphs
when seen as processes indexed by time.
Unfortunately, proof assistants based on intensional type theory
only have modest support for such equational theories,
forcing a prolific use of transports.
This paper reports on the process of mechanizing \citeauthor{YaoPOPL2025}'s results,
comprising the logical relation,
the algebra of computable trajectories with supporting lemmas, and
the fundamental theorem of the logical relation,
in the Rocq theorem prover.
In particular, it describes a systematic treatment of transports
using commuting lemmas.


%% file: sections/introduction.tex
Modern computing applications are increasingly becoming \emph{heterogeneous}, with many systems needing to interface with \emph{foreign objects}.
These foreign objects can range from code that is externally implemented (in a language that may or may not be the same as the one internal to the system)
to actual hardware devices which communicate with the system. 
An example of such an application can be found in a smart home system for monitoring air quality. The controller of the system communicates with different
sensors around the home to gather data on surrounding air temperature, humidity, pressure, etc.
Such sensors are \emph{hardware} devices, and the controller must interact with through a protocol defined
by the manufacturer in a specification. 
These protocols, in addition to prescribing program behavior, can also impose timing constraints on message exchanges between components.
In the example of the air quality system, a sensor may only provide readings after heating up an internal component for some amount of time, after which it also needs time to cool down. 
The other devices in the smart home system must then comply with the time restrictions put forth by the sensor in order to behave according to the specification.

The challenge of verifying such systems against a given protocol specification is twofold: we must address both \textit{(i)} the heterogeneity of the system and \textit{(ii)} the timing constraints imposed by certain devices.
The heterogeneous nature of these applications means there is no common specification language we can use to describe individual components of the system, no common operational model to describe interactions between said components, and no means of composing the verification of individual parts to certify the whole.
The timing constraints mean that each of the concerns put forth by heterogeneity must be solved in a way that ensures timeliness.
Fortunately, the verification framework developed by \citet{YaoPOPL2025} combats each of these issues.

The framework proposed by \citet{YaoPOPL2025} utilizes
\emph{semantic logical relations}
\citep{ConstableBook1986, LoefARTICLE1982},
also known as \emph{semantic typing} \citep{TimanyJACM2024},
to combat the heterogeneity of such systems.
Semantic logical relations permit terms
that are not necessarily (syntactically) well-typed
as inhabitants of the logical relation,
thus accommodating arbitrary foreign objects,
including sensors.
To facilitate verification of \emph{timed message-passing} systems,
the framework uses \textit{session types}
\citep{HondaCONCUR1993,HondaESOP1998,HondaPOPL2008}
as a specification language,
which it enriches with \textit{temporal predicates}.
Session types are behavioral types
\citep{AnconaARITCLE2016,GayRavaraBOOK2017},
which prescribe the order in which messages must be exchanged.
The indexing with a temporal predicate additionally fixes
the point in time/time window at/during which an exchange must happen.
\citet{YaoPOPL2025} specifically chose
\textit{intuitionistic linear logic session types} (ILLST)
\citep{CairesCONCUR2010,ToninhoESOP2013},
which, through their connection to linear logic,
provide a clear distinction between the provider and client of a given protocol.
Thus, the resulting logical relation is indexed with
timed intuitionistic linear logic session types,
referred to as \textit{timed semantic session logical relation} (TSSLR).

The combination of linear session types and temporal predicates enable the prescription of any arbitrary timed message-passing protocol, providing a specification language for a heterogeneous system in the absence of a common one.
Crucial is the use of a \emph{labelled transition system} (LTS) \citep{MilnerBook1980, MilnerBook1999, SangiorgiWalkerBook2001}
to express the message-passing behavior of any program or device in a system. 
In an LTS, each transition on an object is annotated with an \emph{action} that describes the willingness of that object
to engage in communication. 
To accommodate timing, these transitions are also equipped with the point in time at which a message exchange may happen.
However, these time annotations are local descriptions of potential computations---in order to lift them to a global schedule,
\citet{YaoPOPL2025} introduce the notion of a \textit{trajectory}.
A trajectory is a description of how a configuration of processes evolves over time: given an instant in time, a trajectory will produce the configuration of all components at that time.
To validate that a given trajectory is indeed the result of a sequence of timed LTS reductions, \citet{YaoPOPL2025} introduce the notion of a \textit{computable trajectory}, which is a pair of a trajectory and the sequence of reductions which validates it.
The logical relation is then phrased in terms of these computable trajectories, requiring that configurations of processes come equipped with some evidence that they semantically adhere to a given schedule. 



The fact that inhabitants of \citet{YaoPOPL2025}'s logical relation
not only are the actual configurations of processes
but also their associated computable trajectories---%
the evidence that the configuration can step from one time point to another one---%
makes the relation \emph{proof-relevant}
\citep{BentonTLCA2013,BentonPOPL2014,SterlingHarperJACM2021}.
This evidence mirrors the communication structure prescribed by the session type,
giving rise to an \emph{algebra} on trajectories
with operations to \textit{interleave} (combining the trajectories of two concurrently composed process configurations),
\textit{partition} (truncating a trajectory at a given instant),
and \textit{concatenate} (sequentially composing two trajectories).
The definition of these operations necessitates a rich equational theory on computable trajectories,
challenging any mechanization effort in a proof assistant based on intensional type theory.


In this paper, we present the first mechanization of \citet{YaoPOPL2025}'s verification framework,
certifying protocol compliance of timed message-passing heterogeneous systems. 
We carry out the mechanization in the Rocq proof assistant,
comprising the logical relation with computable trajectories and algebras as well as
a timed ILLST system with Fundamental Theorem of the Logical Relation (FTLR),
showing that well-typed ILLST processes inhabit the logical relation.
In contrast to \citet{YaoPOPL2025}'s pen-and-paper proof,
our framework is typed and parametric in the choice of a language defining an inhabitant of the logical relation.

\paragraph{Contributions.}
Our main contributions are
\begin{itemize}
\item A mechanization of a proof-relevant logical relation for protocol compliance of timed message-passing systems;
\item A formalization of a timed ILLST system with soundness proof (FTLR);
\item A formalization of computable trajectories with commuting lemmas to combat the prolific use of transports.
\end{itemize}

\paragraph{Artifact}
All our results were carried out in the Rocq proof assistant,
provided as an artifact with this submission.

%% file: sections/verification-framework.tex
In this section, we go through the definition of the verification framework as implemented in the mechanization.
Unlike \citet{YaoPOPL2025}'s formalization,
our mechanization is typed and parametric in a process language.
Our mechanization relies on a few admitted axioms,
which we refer to in this paper as assumptions.

\subsection{Overview}

The verification framework depends on the definition of three critical components:
the \emph{computational model}, which is defined with a labelled transition system (LTS), the \emph{protocol specification language} of
timed types for message-passing objects, and the \emph{logical relation}.
In \Cref{sub:comp-model}, we lay out the details of the labelled transition system and how it expresses both the intention and timing of a communication for a message-passing process.
Using the LTS, we establish the \textit{process language} and how we use it as an abstraction for arbitrary components in a heterogeneous system.
In \Cref{sub:protocol-spec-language}, we describe the types used to represent and verify the timed message-passing protocols.
Finally, in \Cref{sub:logical-relation}, we define the notion of a computable trajectory and the logical relation. 

\subsection{Computational Model}%
\label{sub:comp-model}

In order to provide a unified specification of arbitrary objects in our system, 
our framework defines a \emph{labelled transition system} (LTS) \citep{MilnerBook1980, MilnerBook1999, SangiorgiWalkerBook2001}.
Defining the computational model in terms of an LTS allows one to \emph{abstractly} describe communications between different parties in a system,
making it particularly well-suited for a heterogeneous message-passing setting.
In this LTS, each transition is annotated with an \emph{action}, which describes the readiness of an object to engage in a communication.
An action can be a readiness to \emph{send} a message, a readiness to \emph{receive} a message, and an empty action.
The empty action represents an actual message exchange, which occurs when two objects have \emph{complementary} actions (i.e. a send and a receive)
and are able to proceed with communication.
The transitions are also annotated with the point in time at which the communication takes place. 

\subsubsection{Actions and Time}
We proceed with formally defining the LTS, beginning with \emph{actions}.
Actions rely on the notion of a \emph{channel},
which is an identifier that a client (e.g. a controller) who wishes to communicate with an object (e.g. a sensor) chooses for the object.
Channels thus serve as the object's point of contact and can be part of a message's content, i.e. its \emph{payload}:

\input{report/action}
\input{report/time}

\subsubsection{Atomic Processes}
\input{report/atomic}

\subsubsection{Process Language}
We now define the \emph{process language} structure, which we use to abstract over the components of a heterogeneous system. Each structure
defines a set of \emph{nameless objects} ($\NObj$), which describes the ``terms'' in the language that are able to interact with a supplied channel,
and a \emph{transition relation} ($\,\stepsToObj{\alpha}{T}$) on those terms, describing how named terms transition to (possibly many) other atomic processes according to some action at some time.
The choice to have \emph{nameless} terms at the language level not only simplifies the treatment of identifiers in our development,
but also allows objects to be defined independently of a particular choice of providing channel (i.e. they are \emph{polymorphic} in the providing channel).
The latter aspect is especially apt for our target domain. The names of the components in  
heterogeneous systems often take the form of \emph{addresses}, such as IP addresses. These addresses 
are usually assigned as the system is wired up, once the components themselves have already been programmed. 
Conceptually, the functionality of the components are---and should be---independent of the names of the objects. 
This dictum is reflected by our use of nameless objects, which is inspired by the notion of
\emph{nameless family of configurations} in \citep{YaoPOPL2025}.

The definition of the process language structure relies on the notion of a finite multiset, which is given first:

\input{report/fmset}
\input{report/proc-lang}
\input{report/cfg-rules}

\subsection{Protocol Specification Language}%
\label{sub:protocol-spec-language}
In this section, we define the language used to prescribe protocols for arbitrary timed message-passing systems.
To do this, we use \emph{session types}
\citep{HondaCONCUR1993,HondaESOP1998,HondaPOPL2008},
in particular \textit{intuitionistic linear logic session types} (ILLST)
\citep{CairesCONCUR2010,ToninhoESOP2013}.
ILLST are well suited for describing the behaviors of message-passing systems
because they provide a clear distinction between client-side and provider-side actions.
In order to also account for timing constraints in these protocols, ILLST are additionally equipped with \textit{temporal predicates}, which allow us to describe the time constraints of a protocol in addition to its intended behavior.
The grammar of \textit{timed} ILLST is given below:

\input{report/timed-types}

For each of the above types, the superscript $t.p$ indicates the time at which a communication may occur (i.e. at a time $t$ satisfying the predicate $p$), while the connective itself describes the nature of the communication.
The connectives include constructs for \emph{sequencing} messages and \emph{branching}. Sequencing is expressed with
the types $A_1 ⊸^{t.p} A_2$ and $A_1 ⊗^{t.p} A_2$. The type $A_1 ⊸^{t.p} A_2$ indicates that, after \emph{receiving} a channel of type $A_1$,
the protocol transitions to behaving as type $A_2$. Dually, type $A_1 ⊗^{t.p} A_2$ denotes the \emph{sending} of a
channel of type $A_1$ and proceeding as type $A_2$ afterwards. Branching is expressed by the types $A_1 \lwith^{t.p} A_2$ and $A_1 ⊕^{t.p} A_2$. 
$A_1 \lwith^{t.p} A_2$ \emph{offers} a choice between channels of type $A_1$ and $A_2$ (i.e. it is ready to \emph{receive} channels of either type),
and $A_1 ⊕^{t.p} A_2$ \emph{makes} a choice between channels of type $A_1$ and $A_2$ (i.e. it is ready to \emph{send} channels of either type).
The choice is conveyed by receiving and sending the labels $π_1$ or $π_2$, which are boolean-valued indicators.
The type $1$ denotes the terminal state of a protocol.

We illustrate timed protocol specification using timed ILLST on the example

$$A_1 ⊸^{t_1.t_0 \leq t_1 \leq t_0 + 15} (A_2 ⊗^{t_2.t_2 = t_1 + 10} 1^{t_3.t_2 \leq t_3})$$

\noindent which requires a process to receive a channel of type $A_1$ at some time $t_1$,
which is at most 15 units of time after an initial time $t_0$,
after which the process must send a channel of type $A_2$ at some time $t_2$,
which is exactly 10 units of time after the receipt of the channel,
after which the process terminates.

\subsection{Logical Relation}%
\label{sub:logical-relation}

With both the operational model and specification language in place, we may now proceed with defining the logical relation.
Before we can give the full definition, however, we must introduce the notion of a trajectory, which relates program execution states to points in time.

\subsubsection{Trajectories}
\label{sub:trajectories}

We proceed with the definition of a trajectory, as well as some concrete trajectories.

\begin{defn}[Trajectory]
A \emph{trajectory} $s$ is a function from a time interval to $\Cfg$,
represented with the following type:
\[ s : (T_0 : \TimeT) → T ≤ T_0 → T_0 < T' → \Cfg \]

\noindent Here, $T$ and $T'$ denote the time interval $[T, T')$, over which the trajectory is defined, where $T_0$ is a point in time within the interval.
For brevity, when we apply a trajectory to a time $T$,
we implicitly assume that the proofs for the bound conditions are also passed and denote the application as simply $s[T]$.
We also define $\trajTy{T}{T'}$ and $[T, T') → \Cfg$ as shorthand for the type $(T_0 : \TimeT) → T ≤ T_0 → T_0 < T' → \Cfg$.
\end{defn}

\begin{defn}[Constant Trajectory]
For configuration $Ω$ and times $T, T'$, the \emph{constant trajectory} $\constTraj{Ω}{T}{T'}$ is defined as the constant function mapping all times to $Ω$.
\end{defn}
\begin{defn}[Trajectory Concat]
For $s_1 : \trajTy{T_1}{T_2}$ and $s_2 : \trajTy{T_2}{T_3}$,
the \emph{concatentation} of trajectories $\concatTraj{s_1}{s_2} : \trajTy{T_1}{T_3}$ is defined by piecing the two trajectories together.
Since time has a decidable strict order, if $T < T_2$ at any time $T$, $\concatTraj{s_1}{s_2}[T]$ returns $s_1[T]$; otherwise, it returns $s_2[T]$.
\end{defn}
\begin{defn}[Trajectory Extension]
For $s : \trajTy{T_1}{T_2}$, $Ω ∈ \Cfg$, and $T_3 ≤ T_1$, we define
$\extendTraj{Ω}{T_3}{s} : \trajTy{T_3}{T_2}$ as the extension of $s$ to start at time $T_3$
by concatenating the constant trajectory $\constTraj{Ω}{T_3}{T_1}$ with $s$.
\end{defn}
\begin{defn}[Trajectory Partition]
For $s : \trajTy{T_1}{T_3}$ and time $T_2$ such that $T_1 ≤ T_2 < T_3$,
we define:
\begin{itemize}
\item $\partBeforeTraj{s}{T_2} : \trajTy{T_1}{T_2}$, the segment of $s$ from time $T_1$ to $T_2$ (i.e. the \emph{before} partition), and
\item $\partAfterTraj{s}{T_2} : \trajTy{T_2}{T_3}$, the segment of $s$ from time $T_2$ to $T_3$ (i.e. the \emph{after} partition).
\end{itemize}
\end{defn}
\begin{defn}[Trajectory Interleave]
For $s_1, s_2 : \trajTy{T}{T'}$, the \emph{interleave} trajectory $\interleaveTraj{s_1}{s_2} : \trajTy{T}{T'}$ is defined by
mapping $T_0$ to $s_1[T_0] \fmconcat{} s_2[T_0]$.
\end{defn}

We then have the following properties on trajectories, which are largely corollaries of function extensionality:
\begin{lem}
For $s ∈ \trajTy{T}{T'}$, the following are true:
\begin{align*}
\extendTraj{Ω}{T}{s} &= s      \\
\extendTraj{Ω}{T_0}{s}[T_1] &= s[T_1] &\text{ for } T ≤ T_1 < T' \\
\extendTraj{Ω}{T_0}{s}[T_1] &= Ω     &\text{ for } T_0 ≤ T_1 < T
\end{align*}
\end{lem}
\begin{lem}
The $\concatTraj{-}{-}$ operator is associative.
\end{lem}

A number of lemmas like the above are proven and are useful for the rest of the development, but we omit them for brevity.
We refer the interested reader to the artifact.

\begin{lem}[Properties of Interleave]
We have the following commuting lemma on the $\interleaveTraj{-}{-}$ operator:

\begin{enumerate}
\item For $s : \trajTy{T_1}{T'}$,
\begin{align*}
   \interleaveTraj{&\constTraj{Ω}{T}{T'}}{\extendTraj{Ω'}{T}{s}}\\
={}&\extendTraj{Ω ⊎ Ω'}{T}{\interleaveTraj{\constTraj{Ω}{T_1}{T'}}{\extendTraj{Ω'}{T_1}{s}}}  
\end{align*}

\item For $s : \trajTy{T_1}{T'}$,
\begin{align*}
   \interleaveTraj{&\extendTraj{Ω'}{T}{s}}{\constTraj{Ω}{T}{T'}} \\
={}&\extendTraj{Ω' ⊎ Ω}{T}{\interleaveTraj{\extendTraj{Ω'}{T_1}{s}}{\constTraj{Ω}{T_1}{T'}}}
\end{align*}

\item For $s_1 : \trajTy{T_1}{T'}$, $s_2 : \trajTy{T_2}{T'}$, and $T_1 ≤ T_2$,
\begin{align*}
  \interleaveTraj{\extendTraj{Ω_1}{T}{s_1}}{&\extendTraj{Ω_2}{T}{s_2}} \\
={}&\extendTraj{Ω_1 ⊎ Ω_2}{T}{\interleaveTraj{s_1}{\extendTraj{Ω_2}{T_1}{s_2}}}
\end{align*}

\item For $s_1 : \trajTy{T_1}{T'}$, $s_2 : \trajTy{T_2}{T'}$, and $T_2 ≤ T_1$,
\begin{align*}
  \interleaveTraj{\extendTraj{Ω_1}{T}{s_1}}{&\extendTraj{Ω_2}{T}{s_2}} \\
={}&\extendTraj{Ω_1 ⊎ Ω_2}{T}{\interleaveTraj{\extendTraj{Ω_1}{T_2}{s_1}}{s_2}}
\end{align*}
\end{enumerate}
\end{lem}

\subsubsection{Computable Trajectories}

\input{report/r-relation}

For the logical relation, we are interested in trajectories that can be \emph{realized} by a configuration multistepping $σ$.
In other words, we want to only consider trajectories which are indeed the results of valid sequences of transitions, which we call \textit{computable trajectories}.
The realizability relation $\mathcal{R}$ is defined inductively in \Cref{fig:r-relation},
using the following proof terms:
\begin{enumerate}
  \item $\RRefl{Ω}$: The constant trajectory $\constTraj{Ω}{T}{T'}$ is validated by the reflexive transition $\CfgRefl{Ω}{T}$,
  \item $\RStepT{r}{T}$: Given trajectory $s$ is validated by $σ$, where $σ$ steps from time $T_2$ to $T_3$,
        the trajectory $\extendTraj{Ω}{T}{s}$ is validated by the stepping $\CfgStepT{Ω}{T}{T_2}{σ}$, the transition which extends $σ$ to start at $T$
  \item $\RStepC{r}{Ω}$: Given trajectory $s$ is validated by $σ$, where $σ$ steps from configuration $Ω'$ to $Ω''$, 
        the trajectory $s$ is also validated by $\CfgStepC{Ω}{Ω'}{T_2}{σ}$, which extends $σ$ with the computation step from $Ω'$ onwards.
\end{enumerate}
\begin{remark}
If $\mathcal{R}(s, σ)$ for $s:\trajTy{T}{T'}$ and
$σ:\CfgStepsOperad{Ω}{T} \CfgSteps{} \CfgStepsOperad{Ω'}{T''}$,
beware of the following facts:
\begin{itemize}
\item $s[T]$, i.e.~the init of $s$, is not necessarily $Ω$. $\RStepC{r}{Ω}$ allows
  insertion of instant steppings at the start of $σ$, while keeping the $s$ unchanged. In essence,
  $s$ describes the \emph{most reduced form} of the stepping.
\item $T'$, the end time of $s$, is not necessarily $T''$. $σ$ is allowed to finish early,
  so we should expect $T'' ≤ T'$.
\end{itemize}
\end{remark}

We have the following lemmas on the $\mathcal{R}$ relation:

\input{report/r-lemmas}

For convenience, we encapsulate $\mathcal{R}$ and its data into the following structure:
\input{report/comp-traj-S}

\subsubsection{Logical Relation}

With the definition of computable trajectories in hand, we may now proceed with defining the logical relation.
The \emph{logical relation} describes the intended runtime behavior of configurations for a given type and point in time,
prescribing what it means for a configuration to be an ``inhabitant'' and thus to comply with the protocol prescribed by the type and imposed temporal predicate.
\input{report/logical-relations}

\cref{fig:logical-relation} gives the definition of the logical relation in terms of
two mutually recursive ``interpretations,'' $\termInterp{A}{T}{d}$ and $\valueInterp{A}{T}{d}$,
often referred to as the ``term'' (or expression, computation) and ``value'' interpretations of type $A$ at time $T$, respectively.
For both interpretations, the annotation $d$ indicates whether the configuration is being classified as a \textit{provider} of protocol $A$ or to be used by a \textit{client} at type $A$ at time $T$, with $d = \withoutStar$ representing the former and $d = \withStar$ the latter.
The inhabitants of the term interpretation are elements $w$ of $\compTraj{T_0}{∞}$, which represent nameless configurations equipped with computable trajectories validating their semantic adherence to a schedule in the interval $[T_0, ∞)$ (for some time $T_0$).
Elements $w$ are in the term interpretation $\termInterp{A}{T}{d}$ if, for all times $T'$ satisfying the temporal predicate at $A$, the configuration in the trajectory at time $T'$ satisfies the value interpretation $\valueInterp{A}{T'}{d}$.
Depending on the modality $d$, one will either need to show $T \leq T'$ as a premise ($d = \withStar$) or conclusion ($d = \withoutStar$).
The inhabitants of the value interpretation are nameless configurations $\nameless{Ω}$, and they are in the value interpretation of $A$ if they are ready to send or receive a message and thus engage in an \emph{external communication}.
The value interpretation is defined by structural induction on the type,
specifying for each type what the expected runtime behavior of an inhabitant is.

In terms of notation,~\citet{YaoPOPL2025} uses $\termInterpYue{A}{T}{\star}$ and $\valueInterpYue{A}{T}{}$
for $\termInterp{A}{T}{\withStar}$ and $\valueInterp{A}{T}{\withoutStar}$, respectively.

We now highlight some notable cases of the value interpretation, explaining how it captures the behavior specified in \cref{sub:protocol-spec-language} at each type.
Suppose we have some $\nameless{Ω} ∈ \NCfg$ and we are checking at mode $d$. Then,
\begin{description}
  \item[\cref{valueInterp:one}] For all channels $a$, $\nameless{Ω}$ is in the value interpretation at the type $1$ and time $T$ if
    $\nameless{Ω}[a]$ sends a closing signal at time $T$ and transitions to the empty set.
  \item[\cref{valueInterp:lolli}]
    For all channels $a$ and $c$, $\nameless{Ω}$ is in the value interpretation at the type $A_1 ⊸ A_2$ at time $T$ if,
    for all computable trajectories $w_1 ∈ \termInterp{A_1}{T}{d^{-1}}$, there exists some configuration $\nameless{Ω}_2$ and computable trajectory $w_2 ∈ \termInterp{A_2}{T}{d}$ such that $\nameless{Ω}[a]$ receives the channel $c$ along $a$ at time $T$ and transitions to $\nameless{Ω}_2[a]$, and the trajectory $w_2$ starts at $\nameless{Ω}_2$ joined with the start of $w_1$ at $c$.
    Here, the notation $d^{-1}$ inverts mode $d$, where $\withoutStar^{-1} = \withStar$ (and vice versa),
    conveying that the received channel $c$ is now being used.
  \item[\cref{valueInterp:lwith}] $\nameless{Ω}$ is in the value interpretation at the type $A_1 \lwith A_2$ at time $T$ if
    for all channels $a$, there exists some $w_1 ∈ \termInterp{A_1}{T}{d}$ such that $\nameless{Ω}[a]$ receives the boolean indicator $π_1$ at time $T$ and transitions to the start of $w_1$ at $a$,
    \emph{and} for all channels $a$, there exists some $w_2 ∈ \termInterp{A_2}{T}{d}$ such that $\nameless{Ω}[a]$ receives the boolean indicator $π_2$ at time $T$ and transitions to the start of $w_2$ at $a$.
\end{description}

The relation is closed under pre-composition of stepping and passage of time:
\begin{lem}[Backwards Closure]\label{lem:back-closure}
For $w_1: \compTraj{T}{T'}$, $w_2: \compTraj{T'}{\infty}$ with $\projEnd(w_1) = \projStart(w_2)$,
then $w_2 ∈ \termInterp{A}{T'}{\withoutStar} \implies \SConcat{w_1}{w_2} \in \termInterp{A}{T}{\withoutStar}$.
\end{lem}
\begin{lem}[Forwards Closure]\label{lem:forwards-closure}
  For $w ∈ \termInterp{A}{T}{\withStar}$, and $T \leq T'$, then
  $\SPartitionA{w}{T'} \in \termInterp{A}{T'}{\withStar}$.
\end{lem}

%% file: report/action.tex
\begin{defn}[Payload]
A \emph{payload} is one of the following:
\begin{itemize}
\item A boolean-valued \emph{selector}, denoted $π_1$ or $π_2$, or
\item A \emph{closing signal} denoted $()$, or
\item A \emph{channel name} $a ∈ \Labels$.
\end{itemize}
\end{defn}

The set $\Labels$ is some countably infinite set of identifiers.

\begin{defn}[Action]
An \emph{action} is one of the following:
\begin{itemize}
\item The constant $ε$, or
\item A triple $(a, d, p)$ where $a ∈ \Labels$ is a channel name,
$d$ is the \emph{direction}, which can be either $!$ (sending) or $?$ (receiving), and $p$ is a payload.
\end{itemize}
Note that the constant $ε$ stands for an empty action (also known as a \emph{silent transition}).
This corresponds to $τ$-transitions in $π$-calculus~\cite{MilnerBook1980,MilnerBook1999,SangiorgiWalkerBook2001}.
We represent the set of all actions as $\Actions$, whose elements are represented with $\alpha$.
\end{defn}

%% file: report/time.tex
The transitions also utilize a notion of time in their annotations, which we define through a set of axioms.
Crucial is the notion of a time \emph{interval} (see \citep{YaoPOPL2025}), which includes both right-bounded $[T, T')$ and
right-unbounded $[T, ∞)$ intervals. We decided to work with an explicitly defined \emph{infinity} time point $∞$ to
represent the unbounded end of a time interval, which is greater than any time point other than itself.

\begin{axiom}[Finite Time]
We assume a set $\FinTimeT$ with a decidable, strict total order $<$.
\end{axiom}
This induces a rich structure on $\FinTimeT$, including the existence of a $\min$ and $\max$ operator and
the propositional irrelevance of the $<$ and $≤$ relation.
\begin{defn}[Time]
We extend $\FinTimeT$ to the set $\TimeT$, whose elements are either an element of $\FinTimeT$ or a constant value $∞$.
We then extend the order $<$ on $\FinTimeT$ to $\TimeT$ by stipulating that $∀ t ∈ \FinTimeT.~t < ∞$
and $¬(∞ < ∞)$.
\end{defn}
We maintain the same notation after extending the definition. We then have the following lemma and fact:
\begin{lem}
The order $<$ on $\TimeT$ is still decidable and is a strict total order.
\end{lem}
\begin{fact}
Natural numbers with the usual $<$ relation form a valid $\FinTimeT$.
\end{fact}

%% file: report/atomic.tex
In our computational model, the smallest unit of computation is an \emph{atomic process},
which can communicate with other processes through channels,
akin to processes in the $\pi$-calculus~\cite{MilnerBook1980,MilnerBook1999,SangiorgiWalkerBook2001}.
So far we have been referring to such a unit colloquially as an ``object.''

\begin{defn}[Atomic Process]
We define a set of \emph{atomic process}, written $\AtomicObj{A}$,
which has the following elements:
\begin{itemize}
\item $\procEx a {\nameless{P}}$, for $a ∈ \Labels$ and $\nameless{P} ∈ A$, or
\item $\procFwd a b$, for $a, b ∈ \Labels$.
\end{itemize}
\end{defn}

Each atomic process has a \emph{providing channel},
which is the process' identifier (as chosen by the client) and acts as a point of contact.
The channel $a$ in both $\procEx a {\nameless{P}}$
and $\procFwd a b$ is the providing channel.
The set $A$ contains process terms parameterized over a providing channel. In this sense,
the elements of $A$ are considered to be ``nameless,'' and become ``named'' once a providing
channel is supplied through the $\procEx a {-}$ construct.

%% file: report/fmset.tex
\begin{axiom}[Finite Multisets (FMSet)]\label{ex:fmset}
We assume the existence of \emph{finite multisets}, defined over a set $X$, denoted $\FMSet{X}$,
to be a finite set with possibly repeated elements drawn from $X$.
We denote it as $\FMSet{X}$, along with a disjoint union operator $\fmconcat$ with unit $\varnothing$.
We assume the commutative monoid laws for $\fmconcat$ and $\varnothing$.
We also assume the ability to create singleton multisets, denoted $\set{x}$ for $x ∈ X$,
and their unions are written as $\set{x_1, x_2, \ldots, x_n}$ for $x_i ∈ X$.
\end{axiom}
\begin{axiom}
We assume the ability to extend a function $f : X → Y$ to a function
$\functorMap{f} : \FMSet{X} → \FMSet{Y}$, defined as applying $f$ to each element in the multiset.
This operation should preserve the monoid structure.
\end{axiom}

%% file: report/proc-lang.tex
\begin{defn}[Process Language]\label{defn:proc-lang}
We define a process language to be a structure with the following data:
\begin{itemize}
\item A set of \emph{nameless objects} $\NObj$, which can be used as the argument of $\AtomicObjPrefix$, and
\item A \emph{transition relation} $\stepsToObj\cdot\cdot$ with the following type:
\[
(\NObj × \Labels) × \Actions × \FinTimeT × \FMSet{\AtomicObj{\NObj}} → \textsf{Prop}
\]
\end{itemize}
We also define the following notation:
\begin{itemize}
\item \ProcLang{} for the set of process language structures,
\item $S ⊢ P \stepsToObj{a!c}{T} Ω$ to say that in the language $S ∈ \ProcLang{}$,
  an object $P$ and a set of atomic objects $Ω$ satisfy the relation
  $\stepsToObj{\cdot}{\cdot}$, with arguments being $(a, !, c)$ and $T$,
\item $S ⊢ P \stepsToObj{a?c}{T} Ω$ to say the same but with direction being $?$,
\item $S ⊢ P \stepsToObj{ε}{T} Ω$ to say the same with the action $ε$.
\end{itemize}
\end{defn}

To represent a \emph{heterogeneous} set of components, we use the type
$\sum_{S ∈ \ProcLang}S.\NObj$. Additionally, we write:
\begin{itemize}
\item $\NObj$ as a shorthand for the type $\sum_{S ∈ \ProcLang}S.\NObj$,
\item $\Obj$ as a shorthand for the type $\AtomicObj{\NObj}$.
\end{itemize}
This $Σ$-quantification of the language means the processes can possibly be from different process languages.

With all of these definitions in place, we can now define the operational model our framework relies on:

\begin{defn}
We introduce the following definitions in order to talk about the execution of processes:
\begin{itemize}
\item The multiset of objects $\Cfg := \FMSet{\Obj}$, called \emph{configurations},
  as well as the nameless version of them, $\NCfg := \Cfg × \NObj$,
\item An overloaded version of multiset union $\fmconcat$ extended to between $\Cfg$ and $\NCfg$:
  $Ω_1 \fmconcat{} ⟨Ω_2, \nameless{P}⟩ := (Ω_1 \fmconcat Ω_2, \nameless{P})$,
\item An instantiation operation $\nameless{Ω}[a]$ for $\nameless{Ω} ∈ \NCfg$ and $a ∈ \Labels$,
  which destructs $\nameless{Ω}$ as $(Ω, \nameless{P})$ and returns $Ω \fmconcat{} \procEx a {\nameless{P}}$,
\item A \textit{transition relation} $\stepsTo{\alpha}{T}$ between $\Cfg$,
  given inductively by the rules shown in~\cref{fig:cfg-rules},
\item A multistep version $\CfgStepsOperad{Ω}{T}\CfgSteps\CfgStepsOperad{Ω'}{T'}$ of the
  transition relation where the action is $ε$, given inductively by the rules shown in~\cref{fig:cfg-multi-rules}.
\end{itemize}
\end{defn}

From the definition, we can see that a nameless configuration has a distinguished \emph{root object},
but once instantiated, the root becomes indistinguishable from the rest.

%% file: report/cfg-rules.tex
\begin{figure}[h!]
\centering
\begin{mathpar}
\inferrule[Step-Obj]
  {S ⊢ P \stepsToObj{α}{T} Ω'}
  {\set{P} \stepsTo{α}{T} \functorMap{λx.(S, x)}(Ω')} \and
\inferrule[Step-Frame]
  {Ω \stepsTo{α}{T} Ω'}
  {Ω \fmconcat Ω_0 \stepsTo{α}{T} Ω' \fmconcat Ω_0} \and
\inferrule[Step-Fwd]{}
  { \Set{ \procEx a {\nameless{P}}, \procFwd b a } \stepsTo{α}{T} \Set{ \procEx b {\nameless{P}} } } \and
\inferrule[Step-Comm]
  {Ω_1 \stepsTo{a!c}{T} Ω_1' \\ Ω_2 \stepsTo{a?c}{T} Ω_2'}
  {Ω_1 \fmconcat Ω_2 \stepsTo{ε}{T} Ω_1' \fmconcat Ω_2'}
\end{mathpar}
\caption{Configuration single-stepping rules}
\label{fig:cfg-rules}
\end{figure}

The transition rules in~\cref{fig:cfg-rules} are as follows:
\begin{enumerate}
\item \textsc{Step-Obj}: Inclusion from language-specific transition,
\item \textsc{Step-Frame}: The \emph{frame rule}, similar to the one in $π$-calculus \cite{MilnerBook1980, MilnerBook1999,SangiorgiWalkerBook2001},
\item \textsc{Step-Fwd}: Reduction of forwarding processes, and
\item \textsc{Step-Comm}: Performs a message exchange.
\end{enumerate}

If a process steps to an empty multiset, we consider it to be a process that \textit{terminates},
and it disappears from the environment.

\begin{figure}[h!]
\centering
\begin{mathpar}
\inferrule{}{\CfgRefl{Ω}{T} : \CfgStepsOperad{Ω}{T} \CfgSteps \CfgStepsOperad{Ω}{T}} \and
\inferrule
  {T_1 ≤ T_2 \\ σ : \CfgStepsOperad{Ω}{T_2} \CfgSteps \CfgStepsOperad{Ω'}{T_3}}
  {\CfgStepT{Ω}{T_1}{T_2}{σ} : \CfgStepsOperad{Ω}{T_1} \CfgSteps \CfgStepsOperad{Ω'}{T_3}} \and
\inferrule
  {Ω \stepsTo{ε}{T_1} Ω' \\ σ : \CfgStepsOperad{Ω'}{T_1} \CfgSteps \CfgStepsOperad{Ω''}{T_2}}
  {\CfgStepC{Ω}{Ω'}{T_1}{σ} : \CfgStepsOperad{Ω}{T_1} \CfgSteps \CfgStepsOperad{Ω''}{T_2}}
\end{mathpar}
\caption{Configuration multistepping rules}
\label{fig:cfg-multi-rules}
\end{figure}

The transition rules in~\cref{fig:cfg-multi-rules} describe the computation of a configuration over time, where
\begin{enumerate}
  \item $\CfgRefl{Ω}{T}$ represents a ``no-op,'' i.e. a transition which does not perform a computation nor progress time,
  \item $\CfgStepT{Ω}{T_1}{T_2}{σ}$ represents a time extension on a transition. Given $σ$ from $T_2$ to $T_3$,
  $\CfgStepT{Ω}{T_1}{T_2}{σ}$ produces a transition which starts at some earlier time $T_1$, with the same starting configuration $Ω$,
  \item $\CfgStepC{Ω}{Ω'}{T_1}{σ}$ represents a progression via computation step. Given $σ$ from $Ω'$ to $Ω''$
  and a step from $Ω$ to $Ω'$,\\ $\CfgStepC{Ω}{Ω'}{T_1}{σ}$ produces a new transition from $Ω$ to $Ω''$ without additionally progressing time.
\end{enumerate}

We have the following admissible operations on multistepping,
which can be handy in constructing multistepping proofs.
They are all essential for the rest of the development.
\begin{lem}[Multistep Frame]
\begin{mathpar}
\inferrule
  {σ : \CfgStepsOperad{Ω}{T} \CfgSteps \CfgStepsOperad{Ω'}{T'}}
  {\CfgStepsFrame{Ω_0}{σ} : \CfgStepsOperad{Ω\fmconcat Ω_0}{T} \CfgSteps \CfgStepsOperad{Ω'\fmconcat Ω_0}{T'}}
\end{mathpar}
\end{lem}
\begin{lem}[Multistep Concat]
\begin{mathpar}
\inferrule
  {σ : \CfgStepsOperad{Ω}{T_1} \CfgSteps \CfgStepsOperad{Ω'}{T_1} \\
  T_1 ≤ T_2 \\
  σ' : \CfgStepsOperad{Ω'}{T_2} \CfgSteps \CfgStepsOperad{Ω''}{T_3}}
  {\CfgStepsConcat{σ}{σ'} : \CfgStepsOperad{Ω}{T_1} \CfgSteps \CfgStepsOperad{Ω''}{T_3}}
\end{mathpar}
\end{lem}
\begin{lem}[\textsf{stepT} on the Right]
\begin{mathpar}
\inferrule
  {σ : \CfgStepsOperad{Ω}{T_1} \CfgSteps \CfgStepsOperad{Ω'}{T_2} \\ T_2 ≤ T_3}
  {\CfgStepTR{Ω}{T_2}{T_2}{σ} : \CfgStepsOperad{Ω}{T_1} \CfgSteps \CfgStepsOperad{Ω'}{T_3}} \and
\end{mathpar}
\end{lem}
\begin{lem}\label{lem:multistep-interleave-lemmas}
The following are true:
\begin{mathpar}
\inferrule
  { σ : \CfgStepsOperad{Ω_1}{T_1} \CfgSteps \CfgStepsOperad{Ω_1'}{T_1'} \\
    σ' : \CfgStepsOperad{Ω_2}{T_2} \CfgSteps \CfgStepsOperad{Ω_2'}{T_2'} \\
    T_1 ≤ T_2 }
  { \CfgStepsInterleaveL{σ}{σ'} : \CfgStepsOperad{Ω\fmconcat Ω_2}{T_1}
      \CfgSteps \CfgStepsOperad{Ω'\fmconcat Ω_2'}{\max(T_1', T_2')} }

\inferrule
  { σ : \CfgStepsOperad{Ω_1}{T_1} \CfgSteps \CfgStepsOperad{Ω_1'}{T_1'} \\
    σ' : \CfgStepsOperad{Ω_2}{T_2} \CfgSteps \CfgStepsOperad{Ω_2'}{T_2'} \\
    T_2 ≤ T_1 }
  { \CfgStepsInterleaveR{σ}{σ'} : \CfgStepsOperad{Ω\fmconcat Ω_2}{T_2}
      \CfgSteps \CfgStepsOperad{Ω'\fmconcat Ω_2'}{\max(T_1', T_2')} }
\end{mathpar}
\end{lem}
\begin{proof}
By mutual-induction. In practice, we have to manually pass the induction hypothesis of the latter lemma
to the former to get the mutual induction to work.
The patched 
\end{proof}
\begin{lem}[Multistep Interleave]\label{lem:multistep-interleave}
\begin{mathpar}
\inferrule
  {σ : \CfgStepsOperad{Ω_1}{T_1} \CfgSteps \CfgStepsOperad{Ω_1'}{T_1'} \\
   σ' : \CfgStepsOperad{Ω_2}{T_2} \CfgSteps \CfgStepsOperad{Ω_2'}{T_2'}}
  {\CfgStepsInterleave{σ}{σ'} : \CfgStepsOperad{Ω\fmconcat Ω_2}{\min(T_1, T_2)}
      \CfgSteps \CfgStepsOperad{Ω'\fmconcat Ω_2'}{\max(T_1', T_2')}}
\end{mathpar}
\end{lem}
\begin{proof}
By~\Cref{lem:multistep-interleave-lemmas}.
\end{proof}
\begin{lem}
If $σ : \CfgStepsOperad{Ω}{T} \CfgSteps \CfgStepsOperad{Ω'}{T'}$, then $T ≤ T'$.
\end{lem}

%% file: report/timed-types.tex
\begin{align*}
A, B, C ::= 1^{t.p} &\mid A_1 ⊸^{t.p} A_2 \mid A_1 ⊗^{t.p} A_2 \\
 &\mid A_1 \lwith^{t.p} A_2 \mid A_1 ⊕^{t.p} A_2
\end{align*}
The notation $t.p$ (which we call the \textit{temporal predicate}) contains a binder $t$ for a time variable and a proposition $p$ on $t$.
The scoping of time variables is best understood through the formation rules of form
$\mathcal{G} ⊢ \isType{A}$, where $\mathcal{G}$ is the context of time variables.
We write $\mathcal{G} ⊢ p : \KwProp$ to mean that $p$ is a well-formed proposition
using the time variables from $\mathcal{G}$.
In~\citet{YaoPOPL2025}, propositions are generated from the usual propositional logic connectives
and time (in)equalities, but our development is not limited to these.

Since all binary connectives have the same scoping rule,
we will only show the rule for $1$ and $⊗$:
\begin{mathpar}
\inferrule{ \mathcal{G}, t ⊢ p : \KwProp }{ \mathcal{G} ⊢ \isType{1^{t.p}} } \and
\inferrule
  { \left( \mathcal{G},t ⊢ \isType{A_i} \right)_{i∈\set{1, 2}} \\\\
    \mathcal{G}, t ⊢ p : \KwProp }
  { \mathcal{G} ⊢ \isType{A_1 ⊗^{t.p} A_2} }
\end{mathpar}
We define the following operator that extracts the proposition from a type
and instantiates it with time $T$. The binary cases are also similar,
so we only show the ones for $1$ and $⊗$:
\begin{align*}
\propOf(1^{t.p}, T) &= p[T/t] \\
\propOf(A_1 ⊗^{t.p} A_2, T) &= p[T/t]
\end{align*}

%% file: report/r-relation.tex
\begin{figure}
\begin{mathpar}
\inferrule{}
  { \RRefl{Ω} : \mathcal{R}(\constTraj{Ω}{T}{T'}, \CfgRefl{Ω}{T}) } \and
\inferrule
  { σ : \CfgStepsOperad{Ω}{T_2} \CfgSteps{} \CfgStepsOperad{Ω'}{T_3} \\
    r : \mathcal{R}(s, σ) }
  { \RStepT{r}{T} : \mathcal{R}(\extendTraj{Ω}{T}{s}, \CfgStepT{Ω}{T}{T_2}{σ}) } \and
\inferrule
  { σ : \CfgStepsOperad{Ω'}{T} \CfgSteps{} \CfgStepsOperad{Ω''}{T_2} \\
    r : \mathcal{R}(s, σ) }
  { \RStepC{r}{Ω} : \mathcal{R}(s, \CfgStepC{Ω}{Ω'}{T}{σ}) }
\end{mathpar}
\caption{Realization relation $\mathcal{R}$}\label{fig:r-relation}
\end{figure}

%% file: report/r-lemmas.tex
\begin{lem}[$\mathcal{R}$-$<$]
If $s : \trajTy{T}{T'}$ is realized, then $T < T'$.
\end{lem}
\begin{lem}[$\mathcal{R}$-$<'$]
If $σ : \CfgStepsOperad{Ω}{T} \CfgSteps{} \CfgStepsOperad{Ω'}{T'}$ realizes a trajectory, then $T < T'$.
\end{lem}

$\mathcal{R}$ also admits a number of proof rules, which are used in later developments:
\begin{lem}[$\mathcal{R}$-Frame]
\begin{mathpar}
\inferrule
  { r : \mathcal{R}(s, σ) }
  { \RFrame{r}{Ω} : \mathcal{R}(\interleaveTraj{r}{\constTraj{Ω}{T}{T'}}, \CfgStepsFrame{Ω}{σ}) }
\end{mathpar}
\end{lem}
\begin{lem}[$\mathcal{R}$-Concat]
\begin{mathpar}
\inferrule
  { r_1 : \mathcal{R}(s_1, σ_1) \\ r_2 : \mathcal{R}(s_2, σ_2) }
  { \RConcat{r_1}{r_2} : \mathcal{R}(\concatTraj{s_1}{s_2}, \CfgStepsConcat{σ_1}{σ_2}) }
\end{mathpar}
\end{lem}
\begin{lem}[$\mathcal{R}$-Partition] For $s : \trajTy{T_1}{T_3}$,
\begin{mathpar}
\inferrule
  { σ : \CfgStepsOperad{Ω_1}{T_1} \CfgSteps{} \CfgStepsOperad{Ω_3}{T_3'} \\
    r : \mathcal{R}(s, σ) \\ T_1 ≤ T_2 < T_3 }
  { \RPartitionB{r}{T_2} : \mathcal{R}(\partBeforeTraj{s}{T_2}, σ_B) } \\
\inferrule
  { σ : \CfgStepsOperad{Ω_1}{T_1} \CfgSteps{} \CfgStepsOperad{Ω_3}{T_3'} \\
    r : \mathcal{R}(s, σ) \\ T_1 ≤ T_2 < T_3 }
  { \RPartitionA{r}{T_2} : \mathcal{R}(\partAfterTraj{s}{T_2}, σ_A) }
\end{mathpar}
where
\[ σ_B : \CfgStepsOperad{Ω_1}{T_1} \CfgSteps{} \CfgStepsOperad{r[T_2]}{T_2'},~
   σ_A : \CfgStepsOperad{r[T_2]}{T_2} \CfgSteps{} \CfgStepsOperad{Ω_3}{T_3'}
\]
are both $\sum$-quantified (existentially quantified but can be projected)
together with their finishing time.
\end{lem}
\begin{lem}\label{lem:R-interleave-lemmas}
The following are true:
\begin{mathpar}
\inferrule
  { r_1 : \mathcal{R}(s_1, σ_1) \\ r_2 : \mathcal{R}(s_2, σ_2) }
  { \RInterleaveL{r_1}{r_2} : \mathcal{R}(\interleaveTraj{s_1}{s_2}, \CfgStepsInterleaveL{σ_1}{σ_2}) }

  \inferrule
  { r_1 : \mathcal{R}(s_1, σ_1) \\ r_2 : \mathcal{R}(s_2, σ_2) }
  { \RInterleaveR{r_1}{r_2} : \mathcal{R}(\interleaveTraj{s_1}{s_2}, \CfgStepsInterleaveR{σ_1}{σ_2}) }
\end{mathpar}
\end{lem}
\begin{proof}
By mutual-induction and unfolding~\Cref{lem:multistep-interleave-lemmas}.
In practice, we have to manually pass the induction hypothesis of the latter
to the former to get the mutual induction to work.
\end{proof}
\begin{lem}[$\mathcal{R}$-Interleave]\label{lem:R-interleave}
\begin{mathpar}
\inferrule
  { r_1 : \mathcal{R}(s_1, σ_1) \\ r_2 : \mathcal{R}(s_2, σ_2) }
  { \RInterleave{r_1}{r_2} : \mathcal{R}(\interleaveTraj{s_1}{s_2}, \CfgStepsInterleave{σ_1}{σ_2}) }
\end{mathpar}
\end{lem}
\begin{proof}
By~\Cref{lem:R-interleave-lemmas}.
\end{proof}

%% file: report/comp-traj-S.tex
\begin{defn}[Nameless Trajectories]
We extend the notion of trajectories to nameless configurations by replacing the codomain
from $\Cfg$ to $\NCfg$, and overload the notation $\nameless{s}[a]$ for
instantiating the nameless configuration across the entire trajectory.
\end{defn}

All operations on trajectories are also defined on nameless trajectories, and they commute with the instantiation operation. These operations are omitted for brevity.
Unfortunately, the instantiation operator is overloaded with trajectory access at a particular time,
but they can be disambiguated via the type of the argument.

We now define the type $\compTraj{T_1}{T_2}$ of computable trajectories in terms of the relation $\mathcal{R}$.

\begin{defn}
Let $\compTraj{T_1}{T_2}$ denote the following data, separated into two parts:
\begin{enumerate}[leftmargin=5.5mm]
\item The data which are themselves nameless, including:
  \begin{itemize}
  \item $\projStart : \NCfg$ for the starting point,
  \item $\projEnd : \NCfg$ for the finishing point,
  \item A nameless trajectory $\projNtraj : [T_1, T_2) → \NCfg$,
  \end{itemize}
\item The data parameterized by a channel name $a ∈ \Labels$, including:
  \begin{itemize}
  \item A finite time $\projStepTo : \FinTimeT$,
  \item A multistepping
    $\projSteps : \CfgStepsOperad{\projStart[a]}{T_1} \CfgSteps{} \CfgStepsOperad{\projEnd[a]}{\projStepTo}$,
  \item A realization proof $\projRealize : \mathcal{R}(\projNtraj[a], \projSteps)$.
  \end{itemize}
\end{enumerate}
The $\pi$ notations are all projection functions, and they are invoked with or without a channel name depending on the kind of data being projected.
For example, for $w : \compTraj{T_1}{T_2}$, we would write $\projStart(w)$ for the starting configuration of $w$ and $\projSteps(w, a)$ for the multistepping on channel $a$ in $w$.
The latter is invoked with the channel name because it comes from the collection of parameterized data in $\compTraj{T_1}{T_2}$.
Furthermore, we will write $w[T]$ for $\projNtraj(w)[T]$.
\end{defn}

The two-part separation in the definition of $\compTraj{T_1}{T_2}$ arises from the fact that, while the multistepping and $\mathcal{R}$ relations
are defined for named configurations, we want to be able to access nameless configurations, as those are the inhabitants of the logical relation.
Storing the nameless and named data separately allows us to conveniently access these nameless configurations without necessarily needing to provide a channel name.

We now define some operations on computable trajectories which are essential to our development:
\begin{defn}[Constant]
For $\nameless{Ω} : \NCfg$ and times $T, T'$,
define $\SRefl{\nameless{Ω}} : \compTraj{T}{T'}$ such that:\\
$∀ a,~\projRealize(\SRefl{\nameless{Ω}}, a)$ is $\RRefl{\nameless{Ω}[a]}$.
The rest of the fields are immediate.
\end{defn}

\begin{defn}[Partition]We define both of the following:
\begin{itemize}
\item For $w : \compTraj{T_1}{∞}$ and $ T_1 ≤ T_2$,
define $\SPartitionB{w}{T_2} : \compTraj{T_1}{T_2}$ such that:
$\projRealize(\SPartitionB{w}{T_2}, a)$ is $\RPartitionB{\projRealize(w, a)}{T_2}$.

\item For $w : \compTraj{T_1}{T_3}$ and $T_1 ≤ T_2 < T_3$,
define $\SPartitionA{w}{T_2} : \compTraj{T_2}{T_3}$ such that:
$\projRealize(\SPartitionA{w}{T_2}, a)$ is $\RPartitionA{\projRealize(w, a)}{T_2}$.
\end{itemize}
In both definitions, the rest of the fields are immediate.
\end{defn}

\begin{fact}
For $w : \compTraj{T_1}{∞}$, if $T_1 ≤ T < ∞$, then
$\projEnd(\SPartitionB{w}{T}) = w[T]$.\\
Also, $\projEnd(\SPartitionB{w}{∞}) = \projEnd(w)$.
\end{fact}

\begin{fact}
For in-bound $T'$, $\SPartitionB{w}{T}[T'] = w[T']$.
\end{fact}

\begin{fact}
$\projStart(\SPartitionA{w}{T}) = w[T]$.
\end{fact}

\begin{defn}[Concat]
For $w_1 : \compTraj{T_1}{T_2}$ and $w_2 : \compTraj{T_2}{T_3}$ such that
$\projEnd(w_1) = \projStart(w_2)$,
define \[\SConcat{w_1}{w_2} : \compTraj{T_1}{T_3}\] such that
$∀ a,~\projRealize(\SConcat{w_1}{w_2}, a)$ is
\[\RConcat{\projRealize(w_1, a)}{\projRealize(w_2, a)}.\]
\end{defn}

\begin{defn}[Step]
For $u : ∀ a,~\nameless{Ω}[a] \stepsTo{ε}{T} \nameless{Ω'}[a]$ and\\ $T ≤ T'$,
define $\SStep{u} : \compTraj{T}{T'}$ such that
$\projStart(\SStep{u}) = \nameless{Ω}$ and $\projEnd(\SStep{u}) = \nameless{Ω'}$.
\end{defn}

\begin{defn}[Communication]
For\ $\nameless{Ω}_1, \nameless{Ω}_1' : \NCfg$,\\ $Ω_2, Ω_2' : \Cfg$,
channel $a ∈ \Labels$, payload $c$, and $T ≤ T'$, define the following two operations:
\begin{itemize}
\item For $u_1 : ∀ b,~\nameless{Ω}_1[b] \stepsTo{a?c}{T} \nameless{Ω}_1'[b]$
and $u_2 : Ω_2 \stepsTo{a!c}{T} Ω_2'$, define $\SCommTP{u_1}{u_2} : \compTraj{T}{T'}$;

\item For $u_1 : ∀ b,~\nameless{Ω}_1[b] \stepsTo{a!c}{T} \nameless{Ω}_1'[b]$
and $u_2 : Ω_2 \stepsTo{a?c}{T} Ω_2'$, define $\SCommTC{u_1}{u_2} : \compTraj{T}{T'}$;
\end{itemize}
such that for both definitions, their $\projStart$ is $\nameless{Ω}_1 \fmconcat Ω_2$,
and their $\projEnd$ is $\nameless{Ω}_1' \fmconcat Ω_2'$.
\end{defn}

\begin{defn}[Frame]
Given $w : \compTraj{T}{T'}$, define the following two operations:
\begin{itemize}
\item For $Ω : \Cfg$, define \[\SFrame{w}{Ω} : \compTraj{T}{T'}\] such that
$\projStart(\SFrame{w}{Ω}) = Ω \fmconcat \projStart(w)$ and
$\projEnd(\SFrame{w}{Ω}) = Ω \fmconcat \projEnd(w)$.

\item For $\nameless{Ω} : \NCfg$ and $a ∈ \Labels$,
define \[\SFrame{w}{\nameless{Ω}} : \compTraj{T}{T'}\] such that
$\projStart(\SFrame{w}{\nameless{Ω}}) = \nameless{Ω} \fmconcat \projStart(w)[a]$ and
$\projEnd(\SFrame{w}{\nameless{Ω}}) = \nameless{Ω} \fmconcat \projEnd(w)[a]$.
\end{itemize}
\end{defn}

The frame operation is a special case of the following interleaving operation:

\begin{defn}[Interleave]
For $w_1 : \compTraj{T}{T'}$, $w_2 : \compTraj{T}{T'}$, and $a ∈ \Labels$,
define $\SInterleave{w_1}{w_2}{a} : \compTraj{T}{T'}$ such that for all in-bound $T''$,
\[\SInterleave{w_1}{w_2}{a}[T''] = w_1[T''][a] \fmconcat w_2[T''].\]
\end{defn}

%% file: report/logical-relations.tex
\begin{figure*}[h!]
\centering
\begin{minipage}[t]{1\columnwidth}
\begin{tabbing}
$\nameless{Ω} ∈ \valueInterp{A_1 ⊸ A_2}{T}{d}$ \= $\iff$ \= something \kill
\fbox{For $w : \compTraj{T_0}{∞}$,} \\[.5em]
$w ∈ \termInterp{A}{T}{\withStar}$ \> $\iff$ \>
  $∀ T',~\propOf(A, T') ∧ T ≤ T' ⇒ w[T'] ∈ \valueInterp{A}{T'}{\withStar}$ \\[.5em]
$w ∈ \termInterp{A}{T}{\withoutStar}$ \> $\iff$ \>
  $∀ T',~\propOf(A, T') ⇒ T ≤ T' ∧ w[T'] ∈ \valueInterp{A}{T'}{\withoutStar}$ \\[.7em]
$\nameless{Ω} ∈ \valueInterp{A_1 ⊸ A_2}{T}{d}$ \> $\iff$ \>
  $∀i∈\set{1, 2},~(∀a,~∃ w_i ∈ \termInterp{A_i}{T}{d},~\nameless{Ω}[a] \stepsTo{a ? π_i}{T} \projStart(w_i)[a])$, \= \kill
$\nameless{Ω} ∈ \valueInterp{1}{T}{d}$ \> $\iff$ \> $∀a,~\nameless{Ω}[a] \stepsTo{a!()}{T} \varnothing$ \> \tagthisline\label{valueInterp:one} \\
$\nameless{Ω} ∈ \valueInterp{A_1 ⊗ A_2}{T}{d}$ \> $\iff$ \> $∀a,~∃c,~∃w_1 ∈ \termInterp{A_1}{T}{d},~∃w_2 ∈ \termInterp{A_2}{T}{d}$, \\
  \>\> $\nameless{Ω}[a] \stepsTo{a!c}{T} \projStart(w_1)[c] \fmconcat{} \projStart(w_1)[a]$ \> \tagthisline\label{valueInterp:tensor} \\
$\nameless{Ω} ∈ \valueInterp{A_1 ⊸ A_2}{T}{d}$ \> $\iff$ \> $∀a,~∀ w_1 ∈ \termInterp{A_1}{T}{\starInvert{d}},~∀c,~∃\nameless{Ω}_2,~∃ w_2 ∈ \termInterp{A_2}{T}{d}$, \\
  \>\> $\nameless{Ω}[a] \stepsTo{a?c}{T} \nameless{Ω}_2[a] ∧ \projStart(w_2) = \nameless{Ω}_2 ⊎ \projStart(w_1)[c]$ \> \tagthisline\label{valueInterp:lolli} \\
$\nameless{Ω} ∈ \valueInterp{A_1 \lwith A_2}{T}{d}$ \> $\iff$ \>
$∀i∈\set{1, 2},~(∀a,~∃ w_i ∈ \termInterp{A_i}{T}{d},~\nameless{Ω}[a] \stepsTo{a ? π_i}{T} \projStart(w_i)[a])$ \> \tagthisline\label{valueInterp:lwith} \\
$\nameless{Ω} ∈ \valueInterp{A_1 ⊕ A_2}{T}{d}$ \> $\iff$ \>
$∃i∈\set{1, 2},~(∀a,~∃ w_i ∈ \termInterp{A_i}{T}{d},~\nameless{Ω}[a] \stepsTo{a ! π_i}{T} \projStart(w_i)[a])$ \> \tagthisline\label{valueInterp:lplus}
\end{tabbing}
\end{minipage}
\caption{Definition of logical relation}
\label{fig:logical-relation}
\end{figure*}

%% file: sections/fundamental-theorem.tex
In this section, we verify arbitrary well-typed programs in a given type system using the logical relation.

\subsection{Overview}
To perform this mode of verification, we need to develop a type system that is strong enough to ensure that any
well-typed term behaves as prescribed by the logical relation. The proof of this property is referred to as the \emph{fundamental theorem of the logical relation}
(FTLR). By carrying out the proof of the FTLR ``once-and-for-all,'' per-program verification reduces to a typechecking problem---if it typechecks, we can simply invoke the FTLR to get our desired result.
Additionally, if typechecking is decidable, the FTLR will make program verification automatic.
In \cref{sec:lang-def}, we will introduce the necessary definitions for the language we want to verify, including the terms, types, typing rules, and semantics.
In \cref{sec:ftlr}, we use the above language definitions to define the fundamental theorem for well-typed terms, as well as an adequacy result that allows us to verify that inhabitants of the logical relation behave as expected.

\subsection{Language Definition}
\label{sec:lang-def}

In this exploration, we will consider a language of process terms which can be specified with the types defined in \cref{sub:protocol-spec-language}.
The grammar for this language is defined in \Cref{fig:term-grammar},
with the terms in the left column being constructors for the corresponding behavioral types in the right column.
The grammar contains boolean-valued signals $e$ (with $π_1$ being $\textcode{false}$ and $π_2$ being $\textcode{true}$), symbols $s$ (which can either be channel names $a ∈ \Labels$, where $\Labels$ is a countably infinite set of identifiers, or a variable $x$),
finite points in time $T$, and process terms $M$.
Each of the terms is annotated with either a time $T$ or a temporal predicate $t.p$ (as introduced in \cref{sub:protocol-spec-language}).
Binding occurrences are denoted by $x ⇒ M$.
As detailed in the next section, we will type these process terms using
timed ILLST \cite{CairesCONCUR2010,ToninhoESOP2013,ToninhoPhD2015}.

\input{report/term-grammar}

\subsubsection{Typing Rules}

\input{report/typing-rules}

To type the process terms defined in \Cref{fig:term-grammar} we use \emph{session types} \cite{HondaCONCUR1993,HondaESOP1998,HondaPOPL2008}, and in particular \emph{timed intuitionistic linear logic session types},
as discussed in \Cref{sub:protocol-spec-language}.

To type the process terms in \Cref{fig:term-grammar}, we use a typing judgment of the form
\[\GF \mid Γ ⊢ M \relAt T :: A,\]
which can be read as ``the process term $M$ provides the communication behavior (i.e. ``session'') of type $A$ at time $T$ under typing context $Γ$, temporal variable context $\mathcal{G}$, and assuming truth of propositions in $\mathcal{F}$.'' In this judgment, $Γ$ provides the typing of free variables in $M$,
$\mathcal{G}$ provides the scoping of time variables, and $\mathcal{F}$. Note that the contents of $\mathcal{G}$ are accessible in the contexts $\mathcal{F}$ and $Γ$, the time $T$, and the type $A$.
The inference rules for this judgment are defined in~\Cref{fig:typing-rules}.

The rules in \Cref{fig:typing-rules} are given in a \emph{sequent calculus},
describing the behavior of a session from the point of view of a provider of the session---expressed by so-called \emph{right rules}---as well as from the point of view of a client---expressed by so-called \emph{left rules}.
Computationally, the rules can be read bottom-up,
where the type of the conclusion denotes the protocol state of the object \textit{before} the message exchange,
and the type of the premise denotes the protocol state \textit{after} the message exchange.
As a consequence, since the conclusion is asserted at time $T$, the premises must take place at some future time.
The rules are thus in agreement with the behavior specified in \cref{sub:protocol-spec-language} at each type.
For example, in case of the right rule \textsc{$⊕$-Right}, the provider sends either of the labels $π_1$ or $π_2$ at time $T$
and then transitions to offering the session $A_i$ at any time $t$ satisfying the temporal predicate $p$, for $i ∈ \set{1, 2}$.
On the other hand, in case of the left rule \textsc{$⊕$-Left}, the client branches on the received label at time $T$,
continuing with either $M₁$ or $M₂$ at a chosen time $T'$ such that $T'$ satisfies the temporal predicate $p$.

While the typing judgment uses a context $Γ$ similarly to the $λ$-calculus,
the typing context here contains a list of \emph{bindings}, which we represent as a finite map
$Γ : \textsf{var} \partialmap \textsf{Types}$, where
$\textsf{var}$ is the set of variables and $\textsf{Types}$ is the set of types.
We do use the conventional $λ$-calculus notation $Γ, x: A$ for the extension of contexts,
but this should be interpreted as inserting the key-value pair $x \mapsto A$ into the map $Γ$.

The rules \textsc{Cut} and \textsc{ID} both carry an additional premise,
written $A ⋉ A' \relAt T$ and $A ⋊ A' \relAt T$, respectively.
These are the \emph{retyping} relations, which are well-explained in~\citet{YaoPOPL2025}[\S 4.1].
The rules for retyping are defined in~\Cref{fig:retyping-rules}.

\subsubsection{Runtime Process Terms}

At runtime, we assign channel names (i.e.~elements of the set $\Labels$) to variables in the context.
Such an assignment is called a \emph{substitution} $σ : \textsf{var} \partialmap \Labels$ from variables to names.
These are not to be confused with substitutions in the $λ$-calculus,
in that we are only substituting channel names for variables, not other terms with possibly free variables.
As a result, the possibility of incurring capture is ruled out in the first place.

Similar to the $λ$-calculus, we write $σ,b/x$ for extending the substitution $σ$ with the key-value
pair $x \mapsto b$, where $b$ is a channel name.
Additionally, we write $σ(x)$ for obtaining the name assigned to the variable $x$ from $σ$.

We define runtime process terms by applying a substitution to process terms $M$, which is defined in~\Cref{fig:subst}.
Notably, the operation $σ-x$ is used for removing $x$ from the substitution $σ$.

\input{report/subst-hat}

\begin{remark}
As is usual, we define $\hat{σ}$ by structural induction on the untyped process term,
removing any bound variable encountered ($σ-x$).
This choice is pretty standard; alternatively we could have updated the substitution $σ$
by mapping the bound variable to itself.

If we were to consider typed terms instead of raw terms,
we could exploit linearity,
the question comes up how $\hat{σ}$ should handle free variables in our type system.
For instance, consider the typing rule for \textsc{$⊗$-Right}
with scoping-related premises explicitly written out:
\begin{mathpar}
\inferrule[\textsc{$⊗$-Right}]
  { \left( \GtFp \mid Γ_i ⊢ M_i :: A_i \right)_{i ∈ \set{1, 2}} \\ \GtFp ⊢ T ≤ t \\ Γ₁ ∩ Γ₂ = \varnothing }
  { \GF \mid Γ₁ ∪ Γ₂ ⊢ (\sendEx^{t.p}(M_1); M_2) \relAt T :: (A_1 ⊗^{t.p} A_2) }
\end{mathpar}
By linear typing, we know that the contexts $Γ₁$ and $Γ₂$ are disjoint,
which indicates a disjoint splitting of $σ$ exists, say $σ = σ₁ ∪ σ₂$,
where $σ_i$ and $Γ_i$ have the same domain for $i ∈ \set{1, 2}$.
If we were to define $\hat{σ}$ for well-typed terms, then it seems more appropriate to define instead:
\[\hat{σ}(\sendEx^{t.p}(M_1); M_2) = \sendEx^{t.p}(\widehat{σ₁}(M_1)); \widehat{σ₂}(M_2).\]
However, \cref{fig:subst} defines $\hat{σ}$ for this case as:
\[\hat{σ}(\sendEx^{t.p}(M_1); M_2) = \sendEx^{t.p}(\hat{σ}(M_1)); \hat{σ}(M_2).\]
This is much easier to implement, but a few lemmas need to be proved about it.
As of now, we are unclear about which approach is favorable.
\end{remark}

We then have the following lemmas on substitution:
\begin{lem}[Composition of Substitution]\label{lem:subst-comp}
For any term $M$, $x ∉ σ$, and channel name $b ∈ \Labels$:
\[
  \widehat{b/x}(\hat{σ}(M)) = \widehat{σ, b/x}(M)
\]
\end{lem}
\begin{proof}
By induction on the structure of $M$ and expanding the definition of $\hat{σ}$.
\end{proof}

Before proceeding to~\Cref{lem:discard,defn:compl},
we define the following, which extends binary relations to elementwise relations on finite maps.
The goal is to related contexts and substitutions---both are finite maps from variables.

\input{report/map-forall2}

A convenient way to represent the facts that two maps have the same domain is to
say that they are related by any suitable relation, so the above facts can be used to derive
other useful propositions about these maps. For example,
we can define the following lemma for related type context maps $Γ$ and substitution maps $σ$:

\begin{lem}[Discard]\label{lem:discard}
For $\GF \mid Γ ⊢ M \relAt T :: A$ and substitutions $σ, σ'$ such that $σ$ and $Γ$ are related maps,
and $σ$ disjoint with $σ'$, $\widehat{σ \cup σ'}(M) = \hat{σ}(M)$.
\end{lem}
\begin{proof}
By induction on the typing of $\GF \mid Γ ⊢ M \relAt T :: A$ and~\Cref{lem:map_forall2-union}.
\end{proof}

\subsubsection{Stepping Rules}

\begin{defn}
We implement a process language (\Cref{defn:proc-lang}),
with nameless objects $\NObj$ being the terms in~\Cref{fig:term-grammar},
and stepping rules specified in~\Cref{fig:dynamics}.
\end{defn}

\input{report/dynamics}

\subsection{Fundamental Theorem}
\label{sec:ftlr}

Before we can state the FTLR, we need to define closing substitutions for the typing context $Γ$,
that we refer to as \emph{complementary configurations}, $δ$, and implement as a finite map from 
the variables in $Γ$ to values $(a, \nameless{Ω}) ∈ \Labels × \compTraj{T}{∞}$. Each pair $(a, \nameless{Ω})$ is
a computable trajectory with its providing channel $a$, and the $δ$ connects variables from $Γ$ to these trajectories.
At runtime, these variables are substituted with channel names,
so we need a way to guarantee that these connected trajectories are also ``well-behaved'', as dictated by the logical relation:

\begin{defn}\label{defn:compl}
We define the set of \emph{complementary configurations} for a context $Γ$
at a given time $T$ as $δ : \text{vars} \partialmap \Labels × \compTraj{T}{∞}$ such that
$δ$ and $Γ$ are related by the relation $s ∈ \termInterp{A}{T}{\withStar}$, for $(a, s) ∈ δ$ and $A ∈ Γ$.

The proposition that $δ$ \emph{is a complementary configuration for $Γ$} is written as $\typedCompl{Γ}{δ}$.
\end{defn}

Applying the first projection to $δ$ obtains a substitution,
which we will denote by $\substOf{δ}$.
We note here that this operation commutes with insertion, union, and deletion on finite maps, which are all mechanized.

The following definition describes the operation of ``applying'' complementary configurations $δ$
to a nameless atomic process $\nameless{P}$. This operation, in a sense,
``links'' a nameless atomic process with all the configurations it communicates with 
by instantiating each of the $\nameless{Ω}$s in the $δ$ map with their providing channel.

\begin{defn}\label{defn:applyCompl}
We may \emph{apply} complementary configurations $δ$ such that $\typedCompl{Γ}{δ}$,
to $\nameless{P} ∈ \NObj$ by mapping the values $(a, w)$ to $\projStart(w)[a]$,
and union them all into a configuration $Ω$.
We then put this result together with $\nameless{P}$ to obtain $(Ω, \nameless{P}) ∈ \NCfg$.

We denote this operation as $\applyCompl{δ}{\nameless{P}}$ and commutes with map and multiset operations.
\end{defn}

\begin{defn}[InterleaveCompl]\label{defn:S-interleave-compl}
  For times $T, T'$, complementary configurations $δ$, and $w \in \compTraj{T}{T'}$,
  we can \emph{interleave} the entirety of $w$ with $δ$ by iterating $δ$,
  where for each $(a, w') ∈ \Labels × \compTraj{T}{∞}$, we fold them as $\SInterleave{w}{w'}{a}$.
  We call this entire operation $\SInterleaveCompl{w}{δ}$.

  The start of the interleaved computable trajectory should be $\applyCompl{δ}{\nameless{P}}$.
\end{defn}

\begin{lem}[Semantic Retyping]
\begin{enumerate}
\item If $A ⋉ B \relAt T$ and given $w ∈ \termInterp{A}{T}{\withoutStar}$, then $w ∈ \termInterp{B}{T}{\withStar}$;
\item If $A ⋊ B \relAt T$ and given $w ∈ \termInterp{B}{T}{\withStar}$, then $w ∈ \termInterp{A}{T}{\withoutStar}$.
\end{enumerate}
\end{lem}
\begin{proof}
By induction on the retyping derivation.
\end{proof}

\begin{theorem}[FTLR]\label{thm:ftlr}
Suppose $\GF\mid Γ ⊢ M \relAt T :: A$.
Assume for all variables in $\mathcal{G}$ that makes $\mathcal{F}$ true,
and for all $δ$ such that $\typedCompl{Γ}{δ}$, then there exists $w \in \compTraj{T}{∞}$ such that
$\applyCompl{δ}{\hat{\substOf{δ}}(M)} = \projStart(w)$ and
$w \in \termInterp{A}{T}{\withoutStar}$.
\end{theorem}
\begin{proof}
By induction on $\GF\mid Γ ⊢ M \relAt T :: A$.
\end{proof}

The following is our \emph{adequacy} result, which is a way for us to ``check''
that inhabitants of the logical relation indeed have the expected property.

\begin{cor}[Adequacy]\label{cor:adequacy}
For channel name $a$, and times $T \leq T'$, if $\varnothing ⊢ \nameless{P} :: 1$ then there exists 
$w \in \compTraj{T}{∞}$ such that $\projStart(w) = (\varnothing, \nameless{P})$ and $w(T')[a] \stepsToObj{a!\varnothing}{T'} \varnothing$.
\end{cor}
\begin{proof}
By FTLR and definition of term interpretation.
\end{proof}

%% file: report/term-grammar.tex
\begin{figure}[ht]
\begin{align*}
e ::={} & π₁ \mid π₂ \\
s ::={} & a \mid x && \text{variable or channel} \\
T :{}   & \FinTimeT && \text{finite time expression}
\end{align*}
\begin{align*}
M ::={} & \fwdEx^T(s) \mid \letEx^T~x:A ← M₁; M₂ \\
  \mid{} & \sendEx^{t.p}() \mid \recvEx^T_s(); M
    && (1) \\
  \mid{} & \recvEx^{t.p}(x ⇒ M) \mid \sendEx^T_s(M₁); M₂
    && (⊸) \\
  \mid{} & \sendEx^{t.p}(M₁); M₂ \mid \recvEx^T_s(x ⇒ M)
    && (⊗) \\
  \mid{} & \recvEx^{t.p}(π₁ ⇒ M₁ \mid π₂ ⇒ M₂) \mid \sendEx^T_s(e); M
    && (\lwith) \\
  \mid{} & \sendEx^{t.p}(e); M \mid \recvEx^T_s(π₁ ⇒ M₁ \mid π₂ ⇒ M₂)
    && (⊕)
\end{align*}
\caption{Grammar for the process language.}\label{fig:term-grammar}
\end{figure}

%% file: report/typing-rules.tex
\begin{figure*}[h!]
\centering
\begin{mathparpagebreakable}
\inferrule[\textsc{Cut}]
  { \GF \mid Γ₁ ⊢ M₁ \relAt T :: A \\\\
    \GF \mid Γ₂, x : A' ⊢ M₂ \relAt T :: C \\\\
    \GF ⊢ A ⋉ A' \relAt T }
  { \GF \mid Γ₁, Γ₂ ⊢ (\letEx^T~x:A' ← M₁ ; M₂) \relAt T :: C }

\inferrule[\textsc{Fwd}]
  { \GF ⊢ A ⋊ A' \relAt T }
  { \GF \mid x : A ⊢ \fwdEx^T(← x) \relAt T :: A' }

\inferrule[\textsc{1-Right}]
  { \GtFp ⊢ T ≤ t }
  { \GF \mid \varnothing ⊢ \sendEx^{t.p}() \relAt T :: 1^{t.p} }

\inferrule[\textsc{1-Left}]
  { \GF \mid Γ ⊢ M \relAt T' :: C \\
    \GF ⊢ T ≤ T' \\ \GF ⊢ p[T'/t] }
  { \GF \mid Γ, x:1^{t.p} ⊢ (\recvEx^{T'}_x(); M) \relAt T :: C } \\

\inferrule[\textsc{$⊸$-Right}]
  { \GtFp \mid Γ, x : A_1 ⊢ M_2 \relAt t :: A_2 \\\\
    \GtFp ⊢ T ≤ t }
  { \GF \mid Γ ⊢ \recvEx^{t.p}(x ⇒ M_2) \relAt T :: (A_1 ⊸^{t.p} A_2) }

\inferrule[\textsc{$⊸$-Left}]
  { \GF \mid Γ₁ ⊢ M₁ \relAt T' :: A_1[T'/t] \\\\
    \GF \mid Γ₂, x : A₂[T'/t] ⊢ M₂ \relAt T' :: C \\\\
    \GF ⊢ T ≤ T' \\ \GF ⊢ p[T'/t] }
  { \GF \mid Γ₁, Γ₂, x : A_1 ⊸^{t.p} A_2 ⊢ (\sendEx^{T'}_x(M_1) ; M_2) \relAt T :: C }

\inferrule[\textsc{$⊗$-Right}]
  { \left( \GtFp \mid Γ_i ⊢ M_i :: A_i \right)_{i ∈ \set{1, 2}} \\ \GtFp ⊢ T ≤ t }
  { \GF \mid Γ₁, Γ₂ ⊢ (\sendEx^{t.p}(M_1); M_2) \relAt T :: (A_1 ⊗^{t.p} A_2) }

\inferrule[\textsc{$⊗$-Left}] 
  { \GF \mid Γ, x : A_1[T'/t], y : A_2[T'/t] ⊢ M \relAt T' :: C \\\\
    \GF ⊢ T ≤ T' \\ \GF ⊢ p[T'/t] }
  { \GF \mid Γ, x : A_1 ⊗^{t.p} A_2 ⊢ (\recvEx^{T'}_x(y ⇒ M)) \relAt T :: C }

\inferrule[\textsc{$\lwith$-Right}]
  { \left( \GtFp \mid Γ ⊢ M_i :: A_i \right)_{i ∈ \set{1, 2}} \\ \GtFp ⊢ T ≤ t }
  { \GtFp \mid Γ ⊢ \recvEx^{t.p}(π₁ ⇒ M₁ \mid π₂ ⇒ M₂) \relAt T :: (A₁ \lwith{}^{t.p} A₂) }

\inferrule[\textsc{$\lwith$-Left}]
  { \GF \mid Γ, x : A_i[T'/t] ⊢ M \relAt T' :: C \\ i ∈ \set{1, 2} \\\\
    \GF ⊢ T ≤ T' \\ \GF ⊢ p[T'/t] }
  { \GF \mid Γ, x : A₁ \lwith{}^{t.p} A₂ ⊢ (\sendEx^{T'}_x(π_i); M) \relAt T :: C }

\inferrule[\textsc{$⊕$-Right}]
  { \GtFp \mid Γ ⊢ M \relAt t :: A_i \\\\
    i ∈ \set{1, 2} \\ \GtFp ⊢ T ≤ t }
  { \GtFp \mid Γ ⊢ (\sendEx^{t.p}(π_i); M) \relAt T :: (A₁ ⊕^{t.p} A₂) }

\inferrule[\textsc{$⊕$-Left}]
  { \left( \GF \mid Γ, x : A_i[T'/t] ⊢ M_i \relAt T' :: C \right)_{i ∈ \set{1, 2}} \\\\
    \GF ⊢ T ≤ T' \\ \GF ⊢ p[T'/t] }
  { \GF \mid Γ, x : A_1 ⊕^{t.p} A_2 ⊢ \recvEx^{T'}_x(π₁ ⇒ M₁ \mid π₂ ⇒ M₂) \relAt T :: C }

\end{mathparpagebreakable}

\caption{Typing rules for the process language.}%
\label{fig:typing-rules}
\end{figure*}

\begin{figure*}[h!]
\centering
\begin{mathparpagebreakable}
\inferrule
  { \GtFq ⊢ p \\ \GtFq ⊢ T ≤ t }
  { \GF ⊢ 1^{t.p} ⋊ 1^{t.q} \relAt T }

\inferrule
  { \GtFq ⊢ p \\ \GtFq ⊢ T ≤ t \\\\
    \GtFq ⊢ B_1 ⋊ A_1 \relAt t \\ \GtFq ⊢ A_2 ⋊ B_2 \relAt t }
  { \GF ⊢ A_1 ⊸^{t.p} A_2 ⋊ B_1 ⊸^{t.q} B_2 \relAt T }

\inferrule
  { \GtFq ⊢ p \\ \GtFq ⊢ T ≤ t \\
    \GtFq ⊢ A_1 ⋊ B_1 \relAt t \\ \GtFq ⊢ A_2 ⋊ B_2 \relAt t }
  { \left( \GF ⊢ A_1 \mathrel{R}^{t.p} A_2 ⋊ B_1 \mathrel{R}^{t.q} B_2 \relAt T \right)_{\text{for}~R∈\set{\lwith, ⊕, ⊗}} }

\\\\

\inferrule
  { \mathcal{G}, t; \mathcal{F}, T ≤ t, q ⊢ p }
  { \GF ⊢ 1^{t.p} ⋉ 1^{t.q} \relAt T }

\inferrule
  { \mathcal{G}, t; \mathcal{F}, T ≤ t, q ⊢ p \\\\
    \mathcal{G}, t; \mathcal{F}, T ≤ t, q ⊢ B_1 ⋉ A_1 \relAt t \\
    \mathcal{G}, t; \mathcal{F}, T ≤ t, q ⊢ A_2 ⋉ B_2 \relAt t }
  { \GF ⊢ A_1 ⊸^{t.p} A_2 ⋉ B_1 ⊸^{t.q} B_2 \relAt T }

\inferrule
  { \mathcal{G}, t; \mathcal{F}, T ≤ t, q ⊢ p \\
    \mathcal{G}, t; \mathcal{F}, T ≤ t, q ⊢ A_1 ⋉ B_1 \relAt t \\
    \mathcal{G}, t; \mathcal{F}, T ≤ t, q ⊢ A_2 ⋉ B_2 \relAt t }
  { \left( \GF ⊢ A_1 \mathrel{R}^{t.p} A_2 ⋉ B_1 \mathrel{R}^{t.q} B_2 \relAt T \right)_{\text{for}~R∈\set{\lwith, ⊕, ⊗}} }

\end{mathparpagebreakable}
\caption{Retyping rules.}%
\label{fig:retyping-rules}
\end{figure*}

%% file: report/subst-hat.tex
\begin{figure}[ht]
\begin{tabbing}
$\hat{σ}(\recvEx^{t.p}(y ⇒ M))$ \= $=$ \= whatever \kill
$\hat{σ}(\fwdEx^T(← x))$ \> $=$ \> $\fwdEx^T(← σ(x))$ \\
$\hat{σ}(\letEx^T~x:A ← M₁; M₂)$ \\
 \> $=$ \> $\letEx^T~x:A ← \hat{σ}(M₁); \widehat{σ-x}(M₂)$ \\
$\hat{σ}(\sendEx^{t.p}())$ \> $=$ \> $\sendEx^{t.p}()$ \\
$\hat{σ}(\recvEx^T_x(); M)$ \> $=$ \> $\recvEx^T_{σ(x)}(); \widehat{σ-x}(M)$ \\
$\hat{σ}(\recvEx^{t.p}(y ⇒ M))$ \> $=$ \> $\recvEx^{t.p}(y ⇒ \widehat{σ-y}(M))$ \\
$\hat{σ}(\sendEx^T_x(M_1); M_2)$ \> $=$ \> $\sendEx^T_{σ(x)}(\hat{σ}(M_1)); \hat{σ}(M_2)$ \\
$\hat{σ}(\sendEx^{t.p}(M_1); M_2)$ \> $=$ \> $\sendEx^{t.p}(\hat{σ}(M_1)); \hat{σ}(M_2)$ \\
$\hat{σ}(\recvEx^T_x(y ⇒ M))$ \> $=$ \> $\recvEx^T_{σ(x)}(y ⇒ \widehat{σ-x}(M))$ \\
$\hat{σ}(\sendEx^T_x(π_i); M)$ \> $=$ \> $\sendEx^T_{σ(x)}(π_i); \hat{σ}(M)$ \text{(for $i∈\set{1,2}$)} \\
$\hat{σ}(\sendEx^{t.p}(π_i); M)$ \> $=$ \> $\sendEx^{t.p}(π_i); \hat{σ}(M)$ \text{(for $i∈\set{1,2}$)} \\
$\hat{σ}(\recvEx^{t.p}(π_1 ⇒ M₁ \mid π_2 ⇒ M₂))$ \\
 \> $=$ \> $\recvEx^{t.p}(π_1 ⇒ \hat{σ}(M₁) \mid π_2 ⇒ \hat{σ}(M₂))$ \\
$\hat{σ}(\recvEx^T_x(π_1 ⇒ M₁ \mid π_2 ⇒ M₂))$ \\
 \> $=$ \> $\recvEx^T_{σ(x)}(π_1 ⇒ \hat{σ}(M₁) \mid π_2 ⇒ \hat{σ}(M₂))$ \\
\end{tabbing}
\caption{Substitution of terms in the process language.}\label{fig:subst}
\end{figure}

%% file: report/map-forall2.tex
\begin{defn}\label{defn:map_forall2}
We say two finite maps $m_1, m_2 : X \partialmap Y$ are \emph{related}
by a relation $R : Y × Y → \textcode{Type}$ when for every $x ∈ X$,
one of the following holds:
\begin{itemize}
\item $m_1$ and $m_2$ are both undefined at $x$,
\item $m_1$ and $m_2$ are both defined at $x$ and $R(m_1[x], m_2[x])$ holds.
\end{itemize}
\end{defn}

\begin{defn}\label{defn:map_ops}
We introduce the following notations:
\begin{itemize}
\item For $m : X \partialmap Y$ and $f : Y → Y'$,
we write $f(m)$ for applying $f$ to the values of $m$, without changing the keys.
\item For $m : X \partialmap Y$ and $x : X$,
we write $m - x$ for removing the key $x$ from $m$.
\end{itemize}
\end{defn}

\begin{lem}\label{lem:map_forall2-union}
For disjointed maps $m_1, m_2$, if $m_1 ⊎ m_2$ are related with $σ$, then we can compute
a disjoint decomposition $σ_1, σ_2$ such that $m_i$ are related with $σ_i$ for $i ∈ \set{1, 2}$,
and $σ = σ_1 ⊎ σ_2$.
\end{lem}

%% file: report/dynamics.tex
\begin{figure}[h!]
\centering
\begin{tabbing}
The following stepping rules always apply: \\
$({\recvEx^{t.p}_b(x ⇒ M)}, c)$ \= $\stepsToObj{b?π_2}{T}$ \= whatevs \kill
$({\fwdEx^T(← b)}, a)$ \> $\stepsToObj{ε}{T}$ \> $\procFwd a b {}$ \\
$({\letEx^T~x:A ← M_1; M_2}, a)$ \\
 \> $\stepsToObj{ε}{T}$ \\
 \` $\procEx{b'}{M_1} \fmconcat{} \procEx a {\widehat{b'/x}(M_2)}$, for $b'∈ \Labels$ \\
$({\recvEx^{T}_b();M}, a)$ \> $\stepsToObj{b?()}{T}$ \> $\procEx a M$ \\
$({\recvEx^T_b(x ⇒ M)}, c)$ \> $\stepsToObj{b?a}{T}$ \> $\procEx c {\widehat{a/x}(M)}$ \\
$({\sendEx^{T}_b(M_1);M_2}, a)$ \> $\stepsToObj{b!c}{T}$ \> $\procEx c {M_1} \fmconcat{} \procEx a {M_2}$ \\
$({\recvEx^T_b(π_1 ⇒ M_1 \mid π_2 ⇒ M_2)}, a)$ \\
 \> $\stepsToObj{b?π_i}{T}$ \> $\procEx a {M_i}$, for $i∈\set{1,2}$ \\
$({\sendEx^T_b(π_i);M}, a)$ \> $\stepsToObj{b!π_i}{T}$ \> $\procEx a M$, for $i∈\set{1,2}$ \\
\\
The following stepping rules are under the premise $p[T/t]$,\\
for universally quantified $T$:\\
$({\sendEx^{t.p}()}, a)$ \> $\stepsToObj{a!()}{T}$ \> $\varnothing$ \\
$({\sendEx^{t.p}(M_1);M_2}, a)$ \> $\stepsToObj{a!c}{T}$ \> $\procEx a {M_2} \fmconcat{} \procEx c {M_1}$ \\
$({\recvEx^{t.p}(x ⇒ M)}, c)$ \> $\stepsToObj{a?b}{T}$ \> $\procEx c {\widehat{b/x}(M)}$ \\
$({\sendEx^{t.p}(π_i);M}, a)$ \> $\stepsToObj{a!π_i}{T}$ \> $\procEx a M$, for $i∈\set{1,2}$ \\
$({\recvEx^{t.p}(π_1 ⇒ M_1 \mid π_2 ⇒ M_2)}, a)$ \\
  \> $\stepsToObj{a?π_i}{T}$ \> $\procEx a {M_i}$, for $i∈\set{1,2}$
\end{tabbing}
\caption{Object-level stepping of the process language.}%
\label{fig:dynamics}
\end{figure}

%% file: sections/mechanization.tex
In this section, we discuss aspects of the Rocq mechanization, including some of our decisions with representation, our use (or lack thereof) of Rocq libraries,
and the deviations we made from the original development in \citeauthor{YaoPOPL2025}

\paragraph{Design decisions}

In a typing judgment $\GF \mid Γ ⊢ M \relAt T :: A$, the context $Γ$ is represented as a \textcode{gmap}
from variable names to type.
For the time context $\GF$, we decided to use a \emph{higher-order} representation
to avoid having to deal with the substitution of time variables. For instance, $\GF \mid Γ ⊢ M \relAt T :: A$
is represented as a Rocq function from the time variables in $\mathcal{G}$ and the proofs of propositions in $\mathcal{F}$
to the inductively defined typing judgment $Γ ⊢ M \relAt T :: A$. This style of representing binding constructs
is inspired by Logical Framework (LF)~\cite{LF}.
This frees us from having to prove substitution lemmas and dealing with structural rules for time propositions.

\paragraph{Rocq libraries}

In our development, we heavily utilize the \textcode{gmap} data structure from the \textlibraryname{stdpp} library~\cite{gmap, stdpp},
with the lemma \textcode{map\_first\_key\_ind} being essential for the lemmas related to~\Cref{defn:S-interleave-compl}.

We developed our own version of \textcode{map\_Forall2}~(\Cref{defn:map_forall2}),
which is valued in \textcode{Type} instead of \textcode{Prop}.
We made essential use of this fact to define~\Cref{lem:map_forall2-union}.

We also had to axiomatize multisets (\Cref{ex:fmset}) instead of using \textcode{gmultiset} from \textlibraryname{stdpp}.
This is because \textcode{gmultiset} is only defined for countable types,
but the multiset of configurations $\Cfg$ is defined with elements $\sum_{S ∈ \ProcLang}S.\NObj$.
In order for this type to be countable, each $\ProcLang$ needs to be countable, which would require the stepping relation in $\ProcLang$ to be countable.
The stepping relation is valued in \textcode{Prop},
and being countable means we must encode stepping proofs into positive numbers.
This violates the principle that \textcode{Prop} should not eliminate to \textcode{Type}.
While the multiset we axiomatized cannot be implemented in Rocq, there are other models of type theory which should support it.

\paragraph{Changes from original development}

We have made a number of notable changes to~\citet{YaoPOPL2025}'s development.

The definition corresponding to~\Cref{lem:multistep-interleave} in Appendix B.1 of \citeauthor{YaoPOPL2025}
proceeds by induction on the length of stepping, rather than structural induction supported by Rocq.
We reconstructed the proof using manual mutual induction.
In the proof of~\Cref{lem:R-interleave}, which goes by unfolding~\Cref{lem:multistep-interleave},
we must create the matching lemmas for the mutually inductive sub-definitions as well.

There is a type error in~\Cref{lem:R-interleave},. Recall the rule:
\[
\RInterleave{r_1}{r_2} : \mathcal{R}(\interleaveTraj{s_1}{s_2}, \CfgStepsInterleave{σ_1}{σ_2})
\]
Suppose $s_1, s_2 : \trajTy{T}{T'}$. Then, $\interleaveTraj{s_1}{s_2} : \trajTy{T}{T'}$.
However, $\CfgStepsInterleave{σ_1}{σ_2}$ begins with $\CfgStepsOperad{Ω_1 ⊎ Ω_2}{\min(T, T)}$,
and there is a type mismatch between $T$ and $\min(T, T)$.
To fix the type error, we would have to rewrite over the fact that $T = \min(T, T)$,
and this creates a rewrite expression in the proof goal. We also used rewrites a couple of times
in the construction of~\Cref{lem:multistep-interleave}. All of these rewrites block $δ$-reduction,
so we developed a family of commuting lemmas to move the rewrites around to simplify the proof goals.

To prove~\Cref{thm:ftlr}, we have to develop a number of operations on $\mathcal{R}$ and $S$,
notably~\Cref{defn:S-interleave-compl}, which are all assumed to be true by~\citet{YaoPOPL2025}.
Essentially, we need to extend a binary operation to a multiset.

\citet{YaoPOPL2025} did not define an interface for process languages.
Instead, they have the $\procEx a -$ syntax, which can be plugged in with essentially ``dynamically-typed'' terms.
Our mechanization contributes a static type signature for all
of these definitions, including but not limited to~\Cref{defn:proc-lang}.
We also took a step further on decoupling the logical relation from syntax.
In~\cite{YaoPOPL2025}, the stepping used by the logical relation relies on the transitions
of untyped terms, which are then semantically typed against timed session types.
We decoupled the logical relation from the syntax, making it parametric over
the language, the process terms, and the object-level stepping.
The only restriction we impose is that language-specific stepping happens within a language,
and the inter-language interaction only happens through message-passing.

%% file: sections/related.tex
In \Cref{sec:mechanization} we have already discussed the differences between our mechanization
and its pen-and-paper formalization given in \citep{YaoPOPL2025},
so we focus here on related mechanization work.

\paragraph*{Mechanizations of logical relations for session types.}

\citet{GollamudiCoqPL2025} contributed the first mechanization of a logical relation for session types.
Although the authors developed a semantic logical relation as well,
they solely assert protocol adherence, but no timeliness.
Moreover, the authors' mechanization is not parametric in a process language,
but fixes the term syntax to be the terms of ILLST,
resulting in a rooted tree structure at run-time with scoped channels.
In contrast, our work recognizes the value of decoupling name assignment for a process from its semantic characterization,
and uses nameless objects.

A slightly different approach to mechanization of session-based message passing
has been taken by Actris \citep{KastbergHinrichsenPOPL2020}.
Actris is a higher-order concurrent verification logic
implemented in Iris
\citep{JungJFP2018}
with support for dependent separation protocols.
Given its powerful foundation,
Actris not only allows verifying protocol compliance of a program,
but also full functional correctness.
The approach taken by Actris, however, is quite different from ours.
Whereas we develop a logical relation that is indexed with ILLST types to prescribe protocols,
Actris uses the logical relations built into Iris and expresses protocols in terms of Iris' ghost state.
Moreover, we mechanize a type system and prove it sound by the FTLR.

\paragraph*{Multi-language verification with processes.}  

Our work echos the development in DimSum \citep{SammlerPOPL2023}, based on Iris \citep{JungJFP2018}.
Similarly to our work, DimSum proposes a process framework for 
multi-language semantics and verification. In this framework, different languages are 
separated into different \emph{modules}, which communicate with each other 
through synchronization of \emph{events}. Our work also models different languages as 
separate processes and allows synchronization via message exchange.
%
DimSum is a very flexible framework: it allows the user to define different sets of events for 
different languages. However, this flexibility comes at a cost, 
requiring a user to set up the embedding of events across languages. Given a 
collection of modules, the user must carry out a linkage proof by hand using proof rules. 
In our system, the communication primitives available are fixed across languages. From 
the point of DimSum, we have chosen a fixed, but rather rich, set of events shared by 
all languages. This allows us to (1) establish a fundamental theorem that automates
composition across languages and (2) support higher-order channels, where the name of 
a message-passing entity is passed as data to another entity.



%% file: sections/acknowledge.tex
This material is based upon work supported by the NSF under Grant No. (2211996 and 2442461) and upon work supported by the AFOSR
under award number FA9550-21-1-0385 (Tristan Nguyen, program manager). Any opinions, findings, and conclusions or recommendations expressed in this material are those of the author(s) and do not necessarily reflect the views of
the National Science Foundation or the U.S. Department of Defense.

The authors would like to thank Robbert Krebbers and Jules Jacobs on their valuable feedback and discussions on mechanization.

%% file: paper.bbl

\begin{thebibliography}{25}


\ifx \showCODEN    \undefined \def \showCODEN     #1{\unskip}     \fi
\ifx \showDOI      \undefined \def \showDOI       #1{#1}\fi
\ifx \showISBNx    \undefined \def \showISBNx     #1{\unskip}     \fi
\ifx \showISBNxiii \undefined \def \showISBNxiii  #1{\unskip}     \fi
\ifx \showISSN     \undefined \def \showISSN      #1{\unskip}     \fi
\ifx \showLCCN     \undefined \def \showLCCN      #1{\unskip}     \fi
\ifx \shownote     \undefined \def \shownote      #1{#1}          \fi
\ifx \showarticletitle \undefined \def \showarticletitle #1{#1}   \fi
\ifx \showURL      \undefined \def \showURL       {\relax}        \fi
\providecommand\bibfield[2]{#2}
\providecommand\bibinfo[2]{#2}
\providecommand\natexlab[1]{#1}
\providecommand\showeprint[2][]{arXiv:#2}

\bibitem[Ancona et~al\mbox{.}(2016)]%
        {AnconaARITCLE2016}
\bibfield{author}{\bibinfo{person}{Davide Ancona}, \bibinfo{person}{Viviana
  Bono}, \bibinfo{person}{Mario Bravetti}, \bibinfo{person}{Joana Campos},
  \bibinfo{person}{Giuseppe Castagna}, \bibinfo{person}{Pierre{-}Malo
  Deni{\'{e}}lou}, \bibinfo{person}{Simon~J. Gay}, \bibinfo{person}{Nils
  Gesbert}, \bibinfo{person}{Elena Giachino}, \bibinfo{person}{Raymond Hu},
  \bibinfo{person}{Einar~Broch Johnsen}, \bibinfo{person}{Francisco Martins},
  \bibinfo{person}{Viviana Mascardi}, \bibinfo{person}{Fabrizio Montesi},
  \bibinfo{person}{Rumyana Neykova}, \bibinfo{person}{Nicholas Ng},
  \bibinfo{person}{Luca Padovani}, \bibinfo{person}{Vasco~T. Vasconcelos},
  {and} \bibinfo{person}{Nobuko Yoshida}.} \bibinfo{year}{2016}\natexlab{}.
\newblock \showarticletitle{Behavioral Types in Programming Languages}.
\newblock \bibinfo{journal}{\emph{Foundations and Trends in Programming
  Languages}} \bibinfo{volume}{3}, \bibinfo{number}{2-3}
  (\bibinfo{year}{2016}), \bibinfo{pages}{95--230}.
\newblock
\urldef\tempurl%
\url{https://doi.org/10.1561/2500000031}
\showDOI{\tempurl}


\bibitem[Benton et~al\mbox{.}(2013)]%
        {BentonTLCA2013}
\bibfield{author}{\bibinfo{person}{Nick Benton}, \bibinfo{person}{Martin
  Hofmann}, {and} \bibinfo{person}{Vivek Nigam}.}
  \bibinfo{year}{2013}\natexlab{}.
\newblock \showarticletitle{Proof-Relevant Logical Relations for Name
  Generation}. In \bibinfo{booktitle}{\emph{11th International Conference on
  Typed Lambda Calculi and Applications ({TLCA})}}
  \emph{(\bibinfo{series}{Lecture Notes in Computer Science},
  Vol.~\bibinfo{volume}{7941})}. \bibinfo{publisher}{Springer},
  \bibinfo{pages}{48--60}.
\newblock
\urldef\tempurl%
\url{https://doi.org/10.1007/978-3-642-38946-7\_6}
\showDOI{\tempurl}


\bibitem[Benton et~al\mbox{.}(2014)]%
        {BentonPOPL2014}
\bibfield{author}{\bibinfo{person}{Nick Benton}, \bibinfo{person}{Martin
  Hofmann}, {and} \bibinfo{person}{Vivek Nigam}.}
  \bibinfo{year}{2014}\natexlab{}.
\newblock \showarticletitle{Abstract effects and proof-relevant logical
  relations}. In \bibinfo{booktitle}{\emph{41st Annual {ACM} {SIGPLAN-SIGACT}
  Symposium on Principles of Programming Languages ({POPL})}}.
  \bibinfo{publisher}{{ACM}}, \bibinfo{pages}{619--632}.
\newblock
\urldef\tempurl%
\url{https://doi.org/10.1145/2535838.2535869}
\showDOI{\tempurl}


\bibitem[Caires and Pfenning(2010)]%
        {CairesCONCUR2010}
\bibfield{author}{\bibinfo{person}{Lu{\'{\i}}s Caires} {and}
  \bibinfo{person}{Frank Pfenning}.} \bibinfo{year}{2010}\natexlab{}.
\newblock \showarticletitle{Session Types as Intuitionistic Linear
  Propositions}. In \bibinfo{booktitle}{\emph{21th International Conference onf
  Concurrency Theory ({CONCUR})}} \emph{(\bibinfo{series}{Lecture Notes in
  Computer Science}, Vol.~\bibinfo{volume}{6269})}.
  \bibinfo{publisher}{Springer}, \bibinfo{pages}{222--236}.
\newblock
\urldef\tempurl%
\url{https://doi.org/10.1007/978-3-642-15375-4\_16}
\showDOI{\tempurl}


\bibitem[Constable et~al\mbox{.}(1986)]%
        {ConstableBook1986}
\bibfield{author}{\bibinfo{person}{Robert~L. Constable},
  \bibinfo{person}{Stuart~F. Allen}, \bibinfo{person}{Mark Bromley},
  \bibinfo{person}{Rance Cleaveland}, \bibinfo{person}{J.~F. Cremer},
  \bibinfo{person}{Robert Harper}, \bibinfo{person}{Douglas~J. Howe},
  \bibinfo{person}{Todd~B. Knoblock}, \bibinfo{person}{Nax~Paul Mendler},
  \bibinfo{person}{Prakash Panangaden}, \bibinfo{person}{James~T. Sasaki},
  {and} \bibinfo{person}{Scott~F. Smith}.} \bibinfo{year}{1986}\natexlab{}.
\newblock \bibinfo{booktitle}{\emph{Implementing Mathematics with the {Nuprl}
  Proof Development System}}.
\newblock \bibinfo{publisher}{Prentice Hall}.
\newblock
\showISBNx{978-0-13-451832-9}
\urldef\tempurl%
\url{http://dl.acm.org/citation.cfm?id=10510}
\showURL{%
\tempurl}


\bibitem[Gay and Ravara(2017)]%
        {GayRavaraBOOK2017}
\bibfield{author}{\bibinfo{person}{Simon~J. Gay} {and}
  \bibinfo{person}{Ant{\'{o}}nio Ravara}.} \bibinfo{year}{2017}\natexlab{}.
\newblock \bibinfo{booktitle}{\emph{Behavioural Types: From Theory to Tools}}.
\newblock \bibinfo{publisher}{River Publishers}.
\newblock


\bibitem[Gollamudi et~al\mbox{.}(2025)]%
        {GollamudiCoqPL2025}
\bibfield{author}{\bibinfo{person}{Tarakaram Gollamudi}, \bibinfo{person}{Jules
  Jacobs}, \bibinfo{person}{Yue Yao}, {and} \bibinfo{person}{Stephanie
  Balzer}.} \bibinfo{year}{2025}\natexlab{}.
\newblock \showarticletitle{A Semantic Logical Relation for Termination of
  Intuitionistic Linear Logic Session Types}. In \bibinfo{booktitle}{\emph{11th
  International Workshop on Coq for Programming Languages ({CoqPL})}}.
\newblock


\bibitem[Harper et~al\mbox{.}(1993)]%
        {LF}
\bibfield{author}{\bibinfo{person}{Robert Harper}, \bibinfo{person}{Furio
  Honsell}, {and} \bibinfo{person}{Gordon Plotkin}.}
  \bibinfo{year}{1993}\natexlab{}.
\newblock \showarticletitle{A Framework for Defining Logics}.
\newblock \bibinfo{journal}{\emph{J. ACM}} \bibinfo{volume}{40},
  \bibinfo{number}{1} (\bibinfo{year}{1993}), \bibinfo{pages}{143–184}.
\newblock
\showISSN{0004-5411}
\urldef\tempurl%
\url{https://doi.org/10.1145/138027.138060}
\showDOI{\tempurl}


\bibitem[Hinrichsen et~al\mbox{.}(2020)]%
        {KastbergHinrichsenPOPL2020}
\bibfield{author}{\bibinfo{person}{Jonas~Kastberg Hinrichsen},
  \bibinfo{person}{Jesper Bengtson}, {and} \bibinfo{person}{Robbert Krebbers}.}
  \bibinfo{year}{2020}\natexlab{}.
\newblock \showarticletitle{Actris: session-type based reasoning in separation
  logic}.
\newblock \bibinfo{journal}{\emph{Proceedings of the ACM on Programming
  Languages}} \bibinfo{volume}{4}, \bibinfo{number}{{POPL}}
  (\bibinfo{year}{2020}), \bibinfo{pages}{6:1--6:30}.
\newblock
\urldef\tempurl%
\url{https://doi.org/10.1145/3371074}
\showDOI{\tempurl}


\bibitem[Honda(1993)]%
        {HondaCONCUR1993}
\bibfield{author}{\bibinfo{person}{Kohei Honda}.}
  \bibinfo{year}{1993}\natexlab{}.
\newblock \showarticletitle{Types for Dyadic Interaction}. In
  \bibinfo{booktitle}{\emph{4th International Conference on Concurrency Theory
  ({CONCUR})}} \emph{(\bibinfo{series}{Lecture Notes in Computer Science},
  Vol.~\bibinfo{volume}{715})}. \bibinfo{publisher}{Springer},
  \bibinfo{pages}{509--523}.
\newblock
\urldef\tempurl%
\url{https://doi.org/10.1007/3-540-57208-2\_35}
\showDOI{\tempurl}


\bibitem[Honda et~al\mbox{.}(1998)]%
        {HondaESOP1998}
\bibfield{author}{\bibinfo{person}{Kohei Honda},
  \bibinfo{person}{Vasco~Thudichum Vasconcelos}, {and} \bibinfo{person}{Makoto
  Kubo}.} \bibinfo{year}{1998}\natexlab{}.
\newblock \showarticletitle{Language Primitives and Type Discipline for
  Structured Communication-Based Programming}. In \bibinfo{booktitle}{\emph{7th
  European Symposium on Programming ({ESOP})}} \emph{(\bibinfo{series}{Lecture
  Notes in Computer Science}, Vol.~\bibinfo{volume}{1381})}.
  \bibinfo{publisher}{Springer}, \bibinfo{pages}{122--138}.
\newblock
\urldef\tempurl%
\url{https://doi.org/10.1007/BFb0053567}
\showDOI{\tempurl}


\bibitem[Honda et~al\mbox{.}(2008)]%
        {HondaPOPL2008}
\bibfield{author}{\bibinfo{person}{Kohei Honda}, \bibinfo{person}{Nobuko
  Yoshida}, {and} \bibinfo{person}{Marco Carbone}.}
  \bibinfo{year}{2008}\natexlab{}.
\newblock \showarticletitle{Multiparty Asynchronous Session Types}. In
  \bibinfo{booktitle}{\emph{35th {ACM} {SIGPLAN-SIGACT} Symposium on Principles
  of Programming Languages ({POPL})}}. \bibinfo{publisher}{{ACM}},
  \bibinfo{pages}{273--284}.
\newblock
\urldef\tempurl%
\url{https://doi.org/10.1145/1328438.1328472}
\showDOI{\tempurl}


\bibitem[Jung et~al\mbox{.}(2018)]%
        {JungJFP2018}
\bibfield{author}{\bibinfo{person}{Ralf Jung}, \bibinfo{person}{Robbert
  Krebbers}, \bibinfo{person}{Jacques{-}Henri Jourdan}, \bibinfo{person}{Ales
  Bizjak}, \bibinfo{person}{Lars Birkedal}, {and} \bibinfo{person}{Derek
  Dreyer}.} \bibinfo{year}{2018}\natexlab{}.
\newblock \showarticletitle{Iris from the ground up: {A} modular foundation for
  higher-order concurrent separation logic}.
\newblock \bibinfo{journal}{\emph{Journal of Functional Programming}}
  \bibinfo{volume}{28} (\bibinfo{year}{2018}), \bibinfo{pages}{e20}.
\newblock
\urldef\tempurl%
\url{https://doi.org/10.1017/S0956796818000151}
\showDOI{\tempurl}


\bibitem[Krebbers(2023)]%
        {gmap}
\bibfield{author}{\bibinfo{person}{Robbert Krebbers}.}
  \bibinfo{year}{2023}\natexlab{}.
\newblock \showarticletitle{Efficient, Extensional, and Generic Finite Maps in
  Coq-std++}.
\newblock  (\bibinfo{year}{2023}).
\newblock
\urldef\tempurl%
\url{https://coq-workshop.gitlab.io/2023/abstracts/coq2023_finmap-stdpp.pdf}
\showURL{%
\tempurl}


\bibitem[Martin{-}L{\"{o}}f(1982)]%
        {LoefARTICLE1982}
\bibfield{author}{\bibinfo{person}{Per Martin{-}L{\"{o}}f}.}
  \bibinfo{year}{1982}\natexlab{}.
\newblock \showarticletitle{Constructive Mathematics and Computer Programming}.
\newblock In \bibinfo{booktitle}{\emph{Logic, Methodology and Philosophy of
  Science VI}}. \bibinfo{series}{Studies in Logic and the Foundations of
  Mathematics}, Vol.~\bibinfo{volume}{104}. \bibinfo{publisher}{Elsevier},
  \bibinfo{pages}{153--175}.
\newblock
\showISSN{0049-237X}
\urldef\tempurl%
\url{https://doi.org/10.1016/S0049-237X(09)70189-2}
\showDOI{\tempurl}


\bibitem[Milner(1980)]%
        {MilnerBook1980}
\bibfield{author}{\bibinfo{person}{Robin Milner}.}
  \bibinfo{year}{1980}\natexlab{}.
\newblock \bibinfo{booktitle}{\emph{A Calculus of Communicating Systems}}.
  \bibinfo{series}{Lecture Notes in Computer Science},
  Vol.~\bibinfo{volume}{92}.
\newblock \bibinfo{publisher}{Springer}.
\newblock
\showISBNx{3-540-10235-3}
\urldef\tempurl%
\url{https://doi.org/10.1007/3-540-10235-3}
\showDOI{\tempurl}


\bibitem[Milner(1999)]%
        {MilnerBook1999}
\bibfield{author}{\bibinfo{person}{Robin Milner}.}
  \bibinfo{year}{1999}\natexlab{}.
\newblock \bibinfo{booktitle}{\emph{Communicating and Mobile Systems: the
  $\pi$-calculus}}.
\newblock \bibinfo{publisher}{Cambridge University Press}.
\newblock
\showISBNx{978-0-521-65869-0}


\bibitem[Sammler et~al\mbox{.}(2023)]%
        {SammlerPOPL2023}
\bibfield{author}{\bibinfo{person}{Michael Sammler}, \bibinfo{person}{Simon
  Spies}, \bibinfo{person}{Youngju Song}, \bibinfo{person}{Emanuele D'Osualdo},
  \bibinfo{person}{Robbert Krebbers}, \bibinfo{person}{Deepak Garg}, {and}
  \bibinfo{person}{Derek Dreyer}.} \bibinfo{year}{2023}\natexlab{}.
\newblock \showarticletitle{DimSum: {A} Decentralized Approach to
  Multi-language Semantics and Verification}.
\newblock \bibinfo{journal}{\emph{Proceedings of the ACM on Programming
  Languages}} \bibinfo{volume}{7}, \bibinfo{number}{{POPL}}
  (\bibinfo{year}{2023}), \bibinfo{pages}{775--805}.
\newblock
\urldef\tempurl%
\url{https://doi.org/10.1145/3571220}
\showDOI{\tempurl}


\bibitem[Sangiorgi and Walker(2001)]%
        {SangiorgiWalkerBook2001}
\bibfield{author}{\bibinfo{person}{Davide Sangiorgi} {and}
  \bibinfo{person}{David Walker}.} \bibinfo{year}{2001}\natexlab{}.
\newblock \bibinfo{booktitle}{\emph{The $\pi$-calculus: a Theory of Mobile
  Processes}}.
\newblock \bibinfo{publisher}{Cambridge University Press}.
\newblock


\bibitem[Sterling and Harper(2021)]%
        {SterlingHarperJACM2021}
\bibfield{author}{\bibinfo{person}{Jonathan Sterling} {and}
  \bibinfo{person}{Robert Harper}.} \bibinfo{year}{2021}\natexlab{}.
\newblock \showarticletitle{Logical Relations as Types: Proof-Relevant
  Parametricity for Program Modules}.
\newblock \bibinfo{journal}{\emph{Journal of the {ACM}}} \bibinfo{volume}{68},
  \bibinfo{number}{6} (\bibinfo{year}{2021}), \bibinfo{pages}{41:1--41:47}.
\newblock
\urldef\tempurl%
\url{https://doi.org/10.1145/3474834}
\showDOI{\tempurl}


\bibitem[{The std++ developers and contributors}(2024)]%
        {stdpp}
\bibfield{author}{\bibinfo{person}{{The std++ developers and contributors}}.}
  \bibinfo{year}{2024}\natexlab{}.
\newblock \bibinfo{booktitle}{\emph{Rocq-std++: {{An}} extended "{{Standard
  Library}}" for {{Rocq}}}}.
\newblock
\urldef\tempurl%
\url{https://gitlab.mpi-sws.org/iris/stdpp/}
\showURL{%
\tempurl}


\bibitem[Timany et~al\mbox{.}(2024)]%
        {TimanyJACM2024}
\bibfield{author}{\bibinfo{person}{Amin Timany}, \bibinfo{person}{Robbert
  Krebbers}, \bibinfo{person}{Derek Dreyer}, {and} \bibinfo{person}{Lars
  Birkedal}.} \bibinfo{year}{2024}\natexlab{}.
\newblock \showarticletitle{A Logical Approach to Type Soundness}.
\newblock \bibinfo{journal}{\emph{Journal of the ACM (JACM)}}
  (\bibinfo{year}{2024}).
\newblock
\newblock
\shownote{To appear}.


\bibitem[Toninho(2015)]%
        {ToninhoPhD2015}
\bibfield{author}{\bibinfo{person}{Bernardo Toninho}.}
  \bibinfo{year}{2015}\natexlab{}.
\newblock \emph{\bibinfo{title}{A Logical Foundation for Session-Based
  Concurrent Computation}}.
\newblock \bibinfo{thesistype}{Ph.\,D. Dissertation}. \bibinfo{school}{Carnegie
  Mellon University and New University of Lisbon}.
\newblock


\bibitem[Toninho et~al\mbox{.}(2013)]%
        {ToninhoESOP2013}
\bibfield{author}{\bibinfo{person}{Bernardo Toninho},
  \bibinfo{person}{Lu{\'{\i}}s Caires}, {and} \bibinfo{person}{Frank
  Pfenning}.} \bibinfo{year}{2013}\natexlab{}.
\newblock \showarticletitle{Higher-Order Processes, Functions, and Sessions: A
  Monadic Integration}. In \bibinfo{booktitle}{\emph{22nd European Symposium on
  Programming ({ESOP})}} \emph{(\bibinfo{series}{Lecture Notes in Computer
  Science}, Vol.~\bibinfo{volume}{7792})}. \bibinfo{publisher}{Springer},
  \bibinfo{pages}{350--369}.
\newblock
\urldef\tempurl%
\url{https://doi.org/10.1007/978-3-642-37036-6\_20}
\showDOI{\tempurl}


\bibitem[Yao et~al\mbox{.}(2025)]%
        {YaoPOPL2025}
\bibfield{author}{\bibinfo{person}{Yue Yao}, \bibinfo{person}{Grant Iraci},
  \bibinfo{person}{Cheng-En Chuang}, \bibinfo{person}{Stephanie Balzer}, {and}
  \bibinfo{person}{Lukasz Ziarek}.} \bibinfo{year}{2025}\natexlab{}.
\newblock \showarticletitle{Semantic Logical Relations for Timed
  Message-Passing Protocols}.
\newblock \bibinfo{journal}{\emph{Proceedings of the ACM on Programming
  Languages}} \bibinfo{volume}{9}, \bibinfo{number}{{POPL}}
  (\bibinfo{year}{2025}), \bibinfo{pages}{1750--1781}.
\newblock
\urldef\tempurl%
\url{https://doi.org/10.1145/3704895}
\showDOI{\tempurl}


\end{thebibliography}
